\documentclass[aoas]{imsart}
\usepackage{lmodern}
\RequirePackage{amsthm,amsmath,amsfonts,amssymb}
\RequirePackage[authoryear]{natbib}
\RequirePackage[colorlinks,citecolor=blue,urlcolor=blue]{hyperref}%% uncomment this for coloring bibliography citations and linked URLs
\theoremstyle{definition}
\newtheorem{definition}{Definition}
\startlocaldefs
\newtheorem*{lemmaone}{Lemma 1}
\usepackage{amsthm}
\theoremstyle{plain}
\newtheorem{lemma}{Lemma}
\usepackage{booktabs,siunitx,rotating,adjustbox}
\usepackage[linesnumbered,ruled,vlined]{algorithm2e}
\usepackage{booktabs}
\usepackage{multirow}
\usepackage{array}
\usepackage{tabularx}
\endlocaldefs

\begin{document}

\begin{frontmatter}
%\title{The Right to Remain Silent: Disentangling Absence from Non-Detection in Whale Monitoring via Deep ZIB Models}
\title{A Deep Zero-Inflated Bernoulli Model of North Atlantic Right Whale Presence Using Multisource Oceanographic Observations}

%\title{A sample article title with some additional note\thanksref{T1}}
%\runtitle{???}
%\thankstext{T1}{A sample of additional note to the title.}

\begin{aug}
%%%%%%%%%%%%%%%%%%%%%%%%%%%%%%%%%%%%%%%%%%%%%%%
%% Only one address is permitted per author. %%
%% Only division, organization and e-mail is %%
%% included in the address.                  %%
%% Additional information can be included in %%
 \author[A]{\fnms{Jiaxiang}~\snm{Ji}\ead[label=e1]{jj983@scarletmail.rutgers.edu}},
\author[B]{\fnms{Laura}~\snm{Nazzaro}\ead[label=e2]{nazzaro@marine.rutgers.edu}},
 \author[A]{\fnms{Ahmed~Aziz}~\snm{Ezzat}\corref{cor1}\ead[label=e3]{aziz.ezzat@rutgers.edu}}

% \and
% \author[B]{\fnms{???}~\snm{???}\ead[label=e3]{???@???}}
%%%%%%%%%%%%%%%%%%%%%%%%%%%%%%%%%%%%%%%%%%%%%%
%% Addresses                                %%
%%%%%%%%%%%%%%%%%%%%%%%%%%%%%%%%%%%%%%%%%%%%%%
 \address[A]{%
Industrial \& Systems Engineering, Rutgers University\\
\printead{e1,e3}
}
\address[B]{%
Department of Marine \& Coastal Sciences, Rutgers University\\
\printead{e2}
}

% \address[B]{???\printead[presep={,\ }]{e2,e3}}
\end{aug}

\begin{abstract}
Effective modeling of endangered marine species, such as the North Atlantic Right Whale, is critical for balancing marine conservation with the growing blue economy. Passive acoustic monitoring data collected by autonomous underwater vehicles, together with satellite-derived oceanographic observations, provide opportunities to predict marine species presence given physical oceanographic conditions. Yet, these multi-source data introduce complex statistical challenges such as zero inflation, imperfect detection, and intricate dependence structures. We propose the Deep Zero-Inflated Bernoulli (DeepZIB) model, a deep statistical method that jointly learns latent species presence and detection probabilities while capturing complex habitat relationships from multi-source covariate information. We establish theoretical results on the model's structural properties and conduct simulations to demonstrate its ability to recover underlying parameters and latent presence fields. Application to real-world data on the North Atlantic Right Whale along the U.S. East Coast demonstrates improved predictive performance in capturing the species' dynamic and spatially varying habitat. We demonstrate DeepZIB's ability to generate high-resolution spatio-temporal presence maps to support targeted, risk-aware ocean management, ranging from offshore energy, to fisheries and maritime transport. Beyond marine science, the proposed methodology is broadly applicable to zero-inflated binary data with imperfect detection and a heterogeneous covariate space.
\end{abstract}


\end{frontmatter}
\section{Introduction}\label{sec:intro}
%The blue economy encompasses the sustainable use of ocean resources for economic growth while safeguarding marine ecosystems. It spans a wide range of activities, from fisheries to offshore renewable energy. Among these, offshore wind has rapidly emerged as a cornerstone of national and regional strategies to decarbonize energy systems and stimulate coastal economies. The United States has set an ambitious target to install 30 Gigawatts (GW) of offshore wind capacity by 2030, with the Mid-Atlantic region expected to be the first and largest contributor to this goal.Yet this same region is home to a diverse marine ecosystem comprising several important marine species, including the critically endangered North Atlantic Right Whale (NARW), formally known as Eubalaena glacialis. Given their dwindling populations, concerns about construction noise and increases vessel traffic highlight the urgent need for effective conservation measures to minimize disruption of NARW habitats. [background of blue economy and transit to the marine ecosystem]
The blue economy refers to the sustainable and economic use of ocean-based resources, %while protecting marine ecosystems through effective monitoring and management. It 
encompassing several important industries such as offshore and marine energy, fisheries and aquaculture, and maritime transport. %These industries contribute substantially to both global and national economies %nearly 5\% of global GDP %The global ocean economy 
%has been valued at more than \$24 trillion in assets, and  generates approximately \$2.5 trillion in annual goods and services, equivalent to nearly 5\% of global GDP 
%\citep{GEF2018BlueEconomy}.  
%Recognizing this global significance, national-level assessments show that the blue economy is a major and rapidly growing component of the U.S. economy. 
%Ocean-related industries in the U.S. supported approximately 2.4 million jobs and generated more than \$397 billion in GDP in 2019 alone. This highlights 
%both the economic scale of ocean industries and their increasing relevance for sustainable development \citep{NOAA2021BlueEconomy}. 
%Yet, the rapid growth of the blue economy 
Yet, the development of the blue economy is inseparable from its surrounding marine ecosystem, since critical ecological habitats and industrial offshore activities often occupy the same space. Ensuring their co-existence through %effective monitoring and 
responsible, risk-aware management is therefore essential. %as the long-term viability of ocean industries depends directly on the health and resilience of the surrounding environment.
%While the blue economy offers substantial economic opportunities, it is essential to account for the risks it poses to endangered species whose habitats overlap with areas of industrial development. 
A timely exemplar of this co-existence is the North Atlantic Right Whale (NARW), % formally known as Eubalaena glacialis, 
a critically endangered species with fewer than 400 individuals alive \citep{nraw2021}. 
%At the current mortality rate, NARWs are projected to be functionally extinct by 2040 \citep{ifawextinct}. The threats are primarily linked to human activities. Specifically, acoustic disturbance is a critical concern during construction, while the operational maintenance requires increased vessel activity, thereby raising the risks of vessel strikes
%For example, noise is a major concern during construction, whereas increased vessel traffic for servicing the energy assets is another risk due to potential vessel strikes 
%\citep{noaadeath,noaastrikes} and habitat disruption \citep{kaldellis2016environmental,shadman2021environmental}. 
Increasing blue economy activities pose considerable risks to NARW habitats, including vessel strikes, %from busy maritime lanes, 
entanglement risks from fishing activities, 
and habitat disruption %from 
due to increased vessel traffic.  
%These risks illustrate the urgency of developing effective strategies to track and mitigate risks for this critically endangered species.

Addressing these growing risks requires responsible ocean management decisions, such as when and where to dispatch vessels for offshore asset construction and servicing to minimize habitat disruption and vessel strike risks. 
Such decisions can benefit from reliable, spatially and temporally resolved information on endangered marine species presence. % highlighting the need for effective monitoring strategies.
%Addressing these growing risks requires effective monitoring strategies of NARW habitats.
Traditional monitoring relies on visual sightings from aerial and vessel line-transect surveys \citep{USNavy2023AFTT}, which provide valuable information on whale distribution and behavior, but are limited by weather, visibility, and species availability. Passive acoustic monitoring (PAM) provides an alternative by monitoring whale vocalizations under diverse environmental conditions \citep{fucile2006self}. Its capabilities have recently expanded through autonomous underwater vehicles (AUVs), such as the underwater glider shown in Figure \ref{fig:data}(b), which can collect localized marine mammal detections and granular oceanographic information over extended periods
%shows an underwater glider (a special type of AUV) used in this study for collecting localized NARW presence and habitat data. When equipped with acoustic sensors, these gliders can navigate quietly through targeted areas for long periods of time to detect marine mammal presence and acquire highly granular habitat information about local oceanic conditions 
\citep{dreyfust2022aligning}. %These new data streams generated by glider-based PAM systems and oceanic information represent the type of high-resolution information that underpins the emerging New Blue Economy %This term broadly refers to the responsible development and management of ocean-based resources through the integration of multi-modal data collected from the ocean environment, including PAM observations, satellite-derived products, and other oceanographic measurements 
%\citep{NOAA_BlueEconomy}

Beyond acoustic detections, understanding marine species presence also requires characterizing the surrounding oceanographic environment. Physical oceanographic conditions, including sea surface temperature, salinity, and oceanographic fronts, vary across space and time and may be associated with the distribution of marine species. Such conditions can be characterized through complementary observing platforms. Underwater gliders provide localized, vertically resolved measurements, whereas satellite observations provide broader spatial coverage of physical oceanographic fields. Together, these multi-source data streams provide a more comprehensive characterization of the spatially and temporally varying physical ocean environment.

A primary use of these data streams is to inform species distribution models (SDMs), which aim to characterize how the presence or distribution of a marine species responds to spatial and temporal variations in environmental conditions \citep{miller2010species}. %SDMs play a central role in conservation planning and risk management for marine habitats \citep{miller2010species}. 
However, applying existing SDMs to these emerging data streams introduces two fundamental statistical challenges. The first is \textit{imperfect detection}. PAM systems detect marine mammals only when vocalizing, leading to substantial availability bias when whales are present but silent, consequently leading to inflated non-detections. The second challenge pertains to the \textit{complexity of the covariate space}. %environmental covariates. 
Integrating the highly localized PAM data with exogenous oceanographic covariates (such as satellite-derived measurements) yields a highly heterogeneous covariate space. Moreover, the relationship between the presence of NARWs, which is a highly mobile marine species, and this heterogeneous covariate information is expected to be complex with intricate contextual, spatial, and temporal dependence. Accurately capturing these complexities calls for highly expressive models that can adapt to the variability of the ocean environment and the corresponding response of NARWs. %to such variability. 

Existing modeling approaches tend to address only one side of the problem. Specifically, many existing SDMs can address the imperfect detection issue. For example, density surface models (DSMs) incorporate detection functions to account for sampling biases \citep{Miller2013DSM}, %can adjust the detection probability based on distance sampling to account for perception bias, 
while occupancy models explicitly separate the latent occurrence state from the observation process \citep{mackenzie2017occupancy}. Another relevant class of statistical approaches is zero-inflated modeling, which accommodates observations arising from distinct zero-generating mechanisms. Zero-inflated models have been extensively studied beyond marine science applications, including manufacturing and quality engineering \citep{Lambert1992ZIP, li1999multivariate}, healthcare and biomedical studies \citep{todem2016}. Such applications highlight the broad relevance of zero-inflated modeling across different domains, where observed zeros may arise from multiple underlying processes.

%, whereas zero-inflated models utilize mixture distributions to distinguish between structural absences and missed detections \citep{Lambert1992ZIP}. 
Despite their utility, these approaches typically rely on rigid parametric assumptions %or simple functional forms that are
that may not be commensurate with the complex and dynamic habitat of highly mobile and intelligent marine species such as the NARWs. %. This limits their capacity to capture the complex, nonlinear relationships inherent in marine ecosystem or to adapt to highly dynamic, non-stationary oceanographic processes. 
Recently, machine learning (ML) methods have been increasingly used as SDMs, especially under a heterogeneous, multi-modal covariate space \citep{ji2024machine}. %to capture  complex nonlinear relationships and can integrate multi-source and multi-modal data streams \citep{elith2008working}. 
Yet, these approaches largely lack explicit mechanisms to account for imperfect detection biases and environmental sampling realities. Importantly, they are fairly limited in interpretability and uncertainty quantification, %relative to their statistical counterparts, 
which are %capabilities that are 
key to supporting risk-aware conservation planning and responsible blue economy management. 
%often provide limited interpretability and uncertainty quantification. %Furthermore, mainstream deep learning models often fail to provide probabilistic outputs (uncertainty quantification) or interpretability, leaving a critical gap in their ability to produce reliable and transparent ecological predictions.

%As we argue in this paper, inferring whale presence from glider-based PAM data and satellite product presents a fundamental challenge. 
To fill this gap, we propose a Deep Zero-Inflated Bernoulli (DeepZIB) method for modeling NARW presence conditional on heterogeneous covariate information derived from glider- and satellite-based measurements. %oceanographic, spatial, temporal, and contextual covariates. 
DeepZIB is formulated as a %probabilistic, 
spatio-temporal statistical model that rigorously accounts for zero inflation while embedding deep learning representations to learn latent presence fields and conditional detection probabilities from heterogeneous covariates. %through a EM-based loss function.  %and ecological grounding of the model. %rigorously accounts for zero inflation in glider-based PAM detections. 
%hierarchal statistical model for zero-inflation with 
%a deep learning representation acts on the exogenous covariate information 
%of complex habitat covariates. The statistical ``layer'' governs the observation process and statistically distinguishes non-detections from true absences, thereby retaining the interpretability and ecological grounding of detection-based models. Embedded within the statistical model, 
%operates as a powerful feature extractor by directly acting on 
%including glider- and satellite-based oceanographic covariates, 
%low-dimensional, but  information-rich representations of the complex habitats governing NARW presence. %This representation is embedded within a likelihood-based statistical framework, enabling joint modeling of latent presence and detection processes. 
%In this way, the DeepZIB approach preserves the rigor, interpretation, and probabilistic nature of rigorous statistical approaches, while leveraging the representational capacity of deep learning to model complex habitat relationships.
In this way, DeepZIB retains the environmental interpretability of statistical SDMs while leveraging the expressive power of deep learning to capture complex habitat relationships.  %while accommodating excess zeros arising from imperfect detection. 
We establish theoretical results on DeepZIB's structural properties and further show that it can be viewed as a generalization of a classical zero-inflated Bernoulli model. %We further assess uncertainty in the estimated latent presence probabilities through a bootstrap-based uncertainty analysis. 
Application to modeling NARW presence along the U.S. East Coast, a geographical region with significant blue economy activities, demonstrates improved %model adequacy and predictive 
predictive performance relative to %purely statistical or machine learning approaches. 
existing statistical or machine learning SDMs. In the context of blue economy management, which is the primary motivation of this work, we highlight how DeepZIB's capability to produce high-resolution, spatially and
temporally varying presence maps can be valuable to support targeted,
risk-aware, and responsible management of various offshore industries, ranging
from marine and offshore energy, to fisheries management and maritime transport.

Although motivated by marine species monitoring, we note that the proposed methodology can be broadly applicable to a wider class of problems involving latent binary states, imperfect observations, and multi-source covariate information. Similar data structures arise in a wide range of applications, but perhaps especially in industrial and environmental monitoring. %including healthcare and industrial monitoring. 
By explicitly separating the latent state process from the observation process, while allowing both components to depend flexibly on high-dimensional covariates, DeepZIB provides a general modeling framework for such complex settings %Its applicability therefore extends beyond ecological modeling to other problems 
characterized by excess zeros, imperfect detection, and intricate covariate, spatial, and temporal dependencies. 

The remainder of this paper is organized as follows. Section \ref{sec:LR} reviews existing approaches for marine mammal habitat modeling. Section \ref{sec:data_description} describes the multi-source dataset. Section \ref{sec:method} details the formulation of the proposed DeepZIB framework, followed by an examination of its theoretical properties and parameter estimation. Section \ref{sec:results} presents simulation experiments and a case study of NARW presence along the U.S. East Coast. Finally, Section \ref{sec:discussion} summarizes our findings and outlines directions for future research.
\section{Literature Review}\label{sec:LR}

Statistical approaches have long provided the foundation for marine SDMs by explicitly linking species observations to environmental covariates. %Broadly, they can be grouped into regression-based models, density surface models, occupancy models, and zero-inflated models. 
Regression-based models (e.g., logistic regression and generalized additive models) %have been widely used. %are widely used in species distribution modeling. 
%These methods 
relate detections or sightings to environmental and oceanographic covariates through predefined link functions \citep{elith2009species,monsarrat2016spatially}. %While such models are interpretable and have proven effective in many applications, they typically rely on additive and rigid parametric assumptions. As a result, they are often limited in their ability to represent dynamic processes and to capture complex, nonlinear dependencies and interactions across space and time.
%hile these models are interpretable and effective for low-dimensional predictors, their reliance on simple smooth structures limits their ability to fully exploit high-dimensional covariate information about marine habitat.
DSMs extend regression-based frameworks by combining spatial modeling with distance sampling \citep{Buckland2001DistanceSampling, Royle2004AbundanceDistance, Miller2013DSM}. 
\iffalse
In the classical two-stage formulation, %(correction for imperfect detection, then spatial modeling), 
a detection function is first estimated from distance measurements, and then a %GAM
regression model is fitted to segment-level counts to obtain a spatially varying density surface \citep{Miller2013DSM,roberts2016habitat}. %This two-stage procedure %is the most widely used class of DSMs and 
%can be viewed as a regression-based spatial model with detectability correction. To avoid separating detection and spatial modeling, 
Alternatively, one-stage DSMs have been proposed in which parameters of the detection and spatial models are estimated simultaneously \citep{Royle2004AbundanceDistance}.
\fi 
Some models adopt point process formulations, modeling detections as realizations of an underlying spatial intensity thinned by the detection function \citep{yuan2017point}. %Importantly, the limitations of DSMs differ from those of regression-based models. While regression-based models are primarily constrained by their rigid parametric assumptions, DSMs face additional challenges arising from their reliance on distance sampling. In particular, 
Understandably, DSMs require precise distance measurements to estimate detection functions, which are unavailable %or imprecise 
in most glider-based platforms.
%These approaches perform well when species-environment relationships are low-dimensional. However, as ecological datasets become increasingly complex and high-dimensional and distance data are reliable, DSMs face specific challenges in the context of glider-based monitoring. Uniquely, they rely heavily on accurate distance measurements to estimate the detection function, which are often unavailable or imprecise for single-sensor autonomous platforms.

Occupancy models extend regression-based approaches by explicitly separating the ecological state from the observation process \citep{MacKenzie2002, Marescot2020Poaching}. %In a standard occupancy model, each site is associated with a latent Bernoulli variable denoting true presence, whereas repeated surveys at that site yield detection and non-detection observations which are used to identify the detection process \citep{MacKenzie2002}. 
%These models can be structurally equivalent to Hidden Markov Models (HMMs) when the unobservable species occurrence is assumed to evolve as a latent state over time \citep{Marescot2020Poaching}. 
%From a statistical perspective, colonization and extinction in an occupancy model can correspond to transition probabilities of the hidden state in an HMM, while detection probability defines the observation model %\textemdash a formulation that is well established in the HMM literature 
%\citep{Scott2002BayesianHMM,Yau2011BNPHMM}. 
%This framework allows non-detections to be possibly attributed to imperfect detection rather than simply regarding it as true ecological absence. 
They %have an attractive mechanism to 
can account for detection biases, but rely on a site-based, repeated-survey design that does not align with glider-based PAM surveys, since gliders roam continuously over a survey area, rather than revisiting fixed sites. %Importantly, the primary objective of occupancy models is to estimate site-level occupancy probabilities, whereas our objective is to model presence and detection at fine spatio-temporal resolutions.
Zero-inflated models were developed to address the frequent occurrence of excess zeros in ecological datasets. In marine mammal monitoring, zeros can arise either because animals are truly absent or because they are present but not detected (i.e., imperfect detection). Zero-inflated formulations explicitly accommodate this by introducing a latent component that governs whether an observation arises from a structural-zero state, together with a conditional model that governs detections or counts when this state is inactive. 

Depending on the type of response variable, zero-inflated models have been developed in several forms, including zero-inflated Poisson and negative binomial models for count data \citep{Lambert1992ZIP,Hall2000}, as well as zero-inflated Bernoulli formulations for binary outcomes \citep{diop2011maximum,fulton2015mixed}. %Diop2016ZIB,  which have been applied to a wide range of applications, from biomedical data \citep{fulton2015mixed}, to transportation surveys \citep{lee2021validation}, and ecological applications \citep{li2022semiparametricZIB}. 
%In contrast to standard regression models, zero-inflated 
These formulations have been studied in a variety of application settings, including manufacturing and quality engineering \citep{Lambert1992ZIP}, healthcare and biomedical research \citep{todem2016}, and behavioral studies \citep{fulton2015mixed}. Such models allow the zero-generating mechanism to depend on covariates and to differ from the process driving nonzero observations. Unlike occupancy models, they do not require a site-based, repeated-survey design, making them applicable to glider-based PAM data. %As such, zero-inflated models provide a flexible statistical framework for representing the mixture of ecological absences and observational failures in marine mammal data.

%Alternative formulations conceptualize sightings as point events, with inhomogeneous Poisson and spatio-temporal point process models enabling inensity-based inference in continuous space and time. More recently, spatiotemporal generalized lienar mixed models (GLMMs) have been developed to capture residual spatial autocorrelation and dynamic shifts by embedding random fields within regression structures (\cite{anderson2022sdmtmb}). 

%Although these statistical approaches form the methodological backbone of marine mammal species distribution modeling and offer clear ecological interpretability, they share important limitations. They rely heavily on predefined parametric or smooth functional forms, making it difficult to capture nonlinear interactions and complex ecological responses. Moreover, most classical frameworks were developed for visual surveys and did not naturally accommodate the nonlinear, high-dimensional, and multi-source covariates collected by the glider. These limitations have motivated growing interest in machine learning methods, which offer flexible, data-driven representations capable of capturing complex nonlinear structure in high-dimensional environmental predictors.

Despite their utility, standard zero-inflated models are often constrained by rigid parametric assumptions %or simplistic functional forms 
that may inadequately represent the complex and dynamic habitats of highly mobile marine species such as the NARW. %Taken together, these statistical frameworks form the methodological backbone of marine mammal species distribution modeling and provide clear ecological interpretability. However, they primarily rely on rigid parametric assumptions, which can limit their ability to fully exploit the high-dimensional, multi-modal marine ecosystem data. 
This has motivated growing interest in ML approaches, %that reduce reliance on rigid parametric forms and provide a flexible framework for capturing complex nonlinear relationships and interactions.
%Machine learning (ML) methods have become increasingly prominent in species distribution modeling as they can exploit high-dimensional covariate information and capture complex nonlinear interactions. Tree-based methods such as boosted regression trees, have demonstrated strong predictive performance in species distribution modeling across both terrestrial and marine ecosystems \citep{elith2008working,ji2024machine}. Kernel-based approaches such as support vector machines, have demonstrated robust classification performance in ecological settings \citep{drake2006modelling}. More recently, deep learning has been applied to ecological problems, leveraging neural network architectures to identify spatiotemporal patterns and integrate heterogeneous data sources, including satellite imagery, and acoustic recordings \citep{rubbens2023machine,norouzzadeh2018automatically}. These studies demonstrate the potential of ML to extend habitat modeling beyond the constraints of traditional statistical methods.
%Machine learning (ML) methods have become increasingly prominent in species distribution modeling as they can 
due to their ability to learn from high-dimensional and heterogeneous covariate information. %for modeling complex habitat relationships. 
Tree- and kernel-based methods %such as boosted regression trees and support vector machines, 
have demonstrated strong predictive performance in species distribution modeling across both terrestrial and marine ecosystems \citep{elith2008working,ji2024machine}. %Kernel-based approaches have also demonstrated robust classification performance in ecological settings \citep{drake2006modelling}. 
More recently, deep learning methods have emerged as powerful SDMs %in various ecological applications %been applied to SDMs to identify spatiotemporal patterns and better address the complexity of ecological niches 
\citep{Chen2016DeepME,deneu2021convolutional}.
%More recently, deep learning has been applied to ecological problems, leveraging neural network architectures to identify spatiotemporal patterns and integrate heterogeneous data sources, including satellite imagery, and acoustic recordings \citep{rubbens2023machine,norouzzadeh2018automatically}.
%Despite their demonstrated predictive performance, most ML models are treated as predictive black boxes and fall short in terms of interpretability. In response, recent lines of research have explored hybrid formulations in which the deep learning structures are embedded within statistical models, a paradigm referred to as deep statistical learning. Most existing work in this area has been developed in general settings, including probabilistic forecasting, deep state-space models, and latent-variable frameworks \citep{damianou2013deep,krishnan2015deep,rangapuram2018deep,salinas2020deepar}. Collectively, these methods illustrate the feasibility of embedding deep learning architectures inside statistical models. However, such hybrid approaches remain largely unexplored in ecological habitat modeling, and to the best of our knowledge, they have not been used to address imperfect detection in acoustic surveys.
Nevertheless, ML methods, in their standard form, do not have explicit mechanisms to account for ecological sampling realities. %data-generating mechanisms underlying ecological observations. 
%Specifically, these approaches 
They typically assume perfect detection and do not distinguish between true absence and non-detection. %As a result, while these methods offer 
Despite their flexibility and representation capacity, %in leveraging multi-source covariate information for representing highly complex habitat relationships, 
these methods may not be well aligned with the observational processes that characterize ecological monitoring data, especially those collected using glider-based PAM systems.

The above limitations of statistical and ML-based SDMs motivate us to define a hybrid method, based on a deep zero-inflated approach, which combines the strengths of both paradigms. Our approach belongs to the emerging class of \textit{statistical deep learning methods} \citep{rangapuram2018deep,ye2025deepmide}, in which elements of a deep learning architecture are embedded within a statistical approach to perform certain learning tasks. Along these lines, our primary modeling contribution is to introduce a deep zero-inflated model for species distribution modeling that is explicitly designed to account for imperfect detection and complex habitat relationships. %by combining the rigor of zero-inflated statistical models with the flexibility of deep learning approaches. 
We establish key theoretical results regarding the model’s structural properties and subsequently apply it to model NARW presence along the U.S. East Coast, a region where expanding blue economy activities would benefit substantially from accurate, high-resolution mapping of critical marine species distributions.
\section{Data Description}\label{sec:data_description}
This work focuses on the U.S. Mid-Atlantic region, an area that supports a diverse marine ecosystem, including the critically endangered NARW, as well as increased Blue economy activities. %such as offshore wind development, vessel activities and fisheries operations. 
%To support real-time species monitoring and assess environmental conditions, we make use of 
Two disparate sources of data are used: (\textit{i}) autonomous underwater glider observations, providing acoustic-based detections and localized oceanographic sensing; and (\textit{ii}) satellite-based oceanographic observations. Table \ref{tab:covariates} provides an overview of the covariates used in this study with their data sources. %The inputs combine glider-based measurements and satellite-derived products, which differ in type, spatial, and temporal resolution. 
Details of each data source are described below. 
\begin{table}[!htbp]
\centering
\caption{Covariates used in this study, measured by underwater gliders and satellite products.}
\label{tab:covariates}

\begin{tabularx}{\textwidth}{@{}llX@{}}
\toprule
\textbf{Data Source} & \textbf{Covariate} & \textbf{Description} \\
\midrule

\multirow{4}{*}{Glider}
& Temperature ($^\circ$C)
& Water temperature \\

& Salinity (psu)
& Dissolved salts concentration \\

& Oxygen concentration ($\mu$mol/L)
& Amount of dissolved oxygen in the water \\

& Depth (m)
& Vertical distance from sea surface to glider's position \\

\midrule

\multirow{3}{*}{Satellite}
& Sea surface temperature ($^\circ$C)
& Temperature of the ocean's surface layer \\

& Chlorophyll (mg/m$^3$)
& Phytoplankton pigment concentration \\

& Frontal value
& Gradient strength across water masses \\

\bottomrule
\end{tabularx}
\end{table}
\subsection{Glider Datasets}
The glider dataset, $\boldsymbol{\mathcal{Z}}$, used in this study was collected from nine Generation 3 Teledyne Slocum glider missions conducted between August 2020 and June 2022, operating off the %south 
coast of New Jersey between $38.5^\circ N$ to $39.5^\circ N$ and $74.5^\circ W$ to $73.5^\circ W$. Each glider was equipped with a digital acoustic monitoring (DMON) system developed by Woods Hole Oceanographic Institution (WHOI), enabling near real-time passive detection of marine mammals %The system runs a low-frequency classification algorithm onboard and transmits detected events back to shore for expert validation 
\citep{baumgartner2011generalized}. %As the sound source cannot be localized directly, detections are georeferenced to the glider's location at the time the sound was recorded.
Across all deployments, a total of $104$ NARW detections were identified. In addition to acoustic recordings, the glider simultaneously collected profiles of oceanographic variables at a vertical resolution of $0.25$ meters (m), including water temperature, oxygen concentration, salinity, and glider depth. Figure \ref{fig:data}(a) shows an example of the oceanographic data collected by a glider during a single deployment. %As the glider repeatedly dove and climbed through the water column, it recorded vertical profiles of the oceanographic variables. 
\iffalse
\begin{table}[ht]
\centering
\caption{Glider mission information and number of detections per mission.}
\label{tab:trips}
\begin{tabular}{|c|c|}
\hline
\textbf{Trips}                    & \textbf{Number of Detections} \\ \hline
2020-07-29 to 2020-08-26  & 0                   \\ \hline
2020-10-03 to 2020-11-05 & 0                   \\ \hline
2020-11-19 to 2020-12-21 & 47                  \\ \hline
2021-02-08 to 2021-03-08 & 21                  \\ \hline
2021-11-20 to 2021-12-17 & 3                   \\ \hline
2022-01-13 to 2022-02-11 & 17                  \\ \hline
2022-02-15 to 2022-03-16 & 16                  \\ \hline
2022-03-30 to 2022-04-22 & 0                   \\ \hline
2022-05-20 to 2022-06-10 & 0                   \\ \hline
\textbf{Total} & \textbf{104}                   \\ \hline
\end{tabular}
\end{table}
\fi

\subsection{Satellite Datasets}
\begin{figure}[]
\centering
\includegraphics[trim=1.5cm 1.5cm 0cm 3cm,clip,width = .91\linewidth]{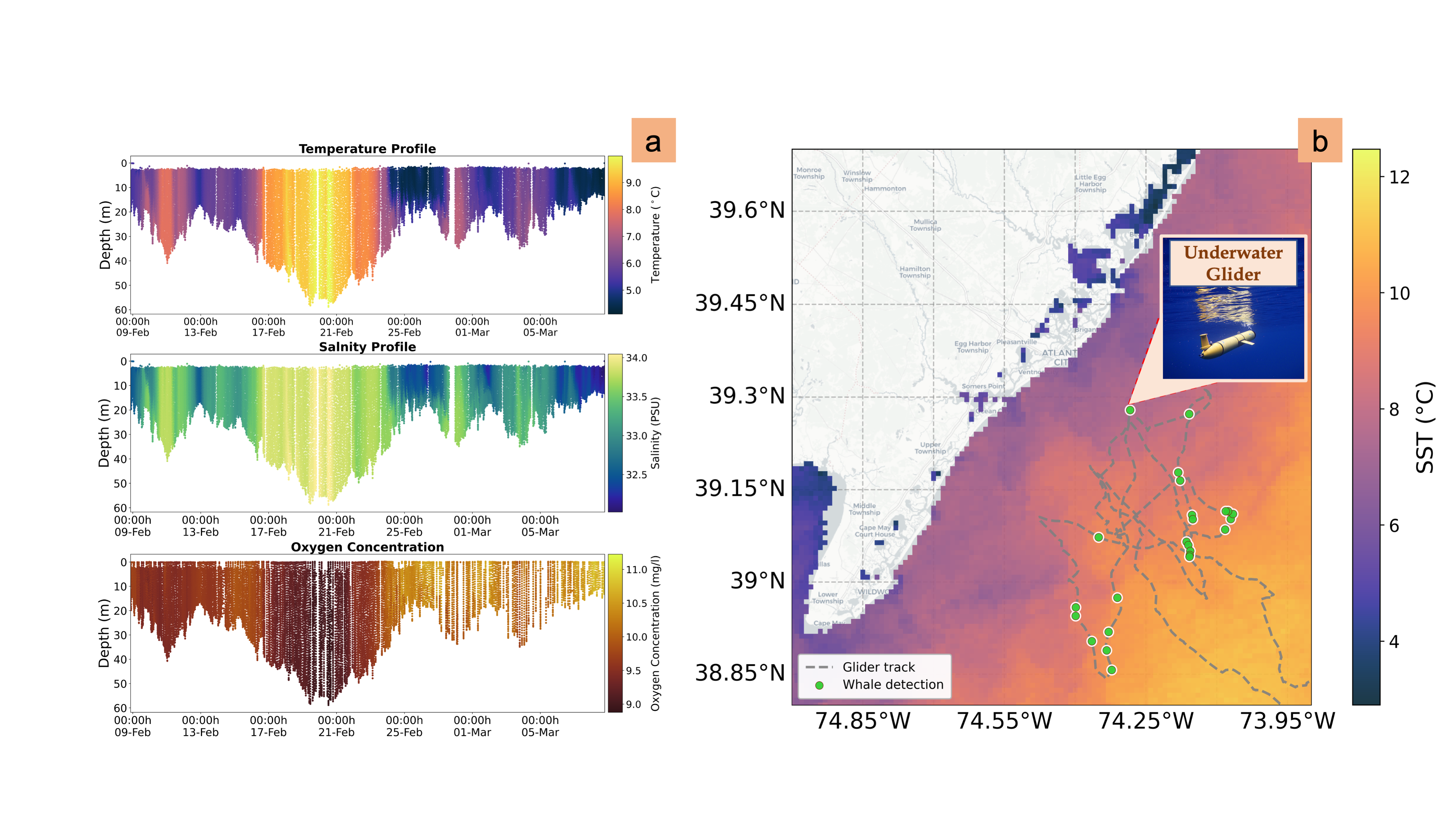}
\caption{Sample visualization of glider and satellite data. (a) glider‐measured %environmental variables 
variables between Feb 8th and Mar 10th, 2021. (b) satellite-based SST %derived from VIIRS Suomi NPP 1-Day 750 m Composite Northwest Atlantic 
on Dec 27th, 2020, with the corresponding glider track from that mission (grey lines) and %a total of 
24 detections (green dots).}
\label{fig:data}
\end{figure}
%Satellite observations complement the glider dataset by providing broad spatial coverage of surface ocean conditions. 
The satellite dataset, $\boldsymbol{\mathcal{X}}$, contains information about key oceanic variables such as sea surface temperature (SST), chlorophyll concentration, as well as the frontal value %which reflects the magnitude of differences between adjacent water masses in SST and ocean color variable space 
\citep{oliver2008objective}. %One challenge with satellite data is %the presence of clouds, which often limits 
%limited spatial coverage. %To mitigate this, 
%Hence, we integrate multiple satellite sources 
To improve data availability, %Specifically, 
we use three satellite sources: (\textit{i}) “NOAA/NESDIS/STAR GHRSST GOES16 SST Daily
Composite SST" (GOES); (\textit{ii}) "VIIRS Suomi NPP 1-Day 750 m Composite Northwest Atlantic" (VIIRS); and (\textit{iii}) "MODIS Aqua 3-Day 1 km Composite Northwest Atlantic" (MODIS). %Since GOES data are likely to be less affected by clouds due to its geostationary nature, it is always our first choice. When GOES data are unavailable for a given location and time, VIIRS is used next as it offers higher spatial resolution compared to MODIS. MODIS is used as a fallback for SST and chlorophyll when neither of the other two sources is available. Since frontal value is only available from MODIS, we solely rely on MODIS for it. 
Figure \ref{fig:data}(b) shows an example of a daily SST map obtained from VIIRS, overlaid with the glider track and co-located detections. %from a co-located (multi-day) mission. %While the satellite image corresponds to a specific day, the glider trajectory spans multiple days. %of the survey and is shown here to illustrate the spatial coverage. 
The inset in Figure \ref{fig:data}(b) shows the glider %used %in our surveys 
for illustration purposes.

%Statistical modeling of such data therefore requires accommodating heterogeneous covariates while explicitly accounting for imperfect detection.
\section{Methodology}
\label{sec:method}
%\subsection{Statistical setup}
We consider the statistical problem of inferring species presence from imperfect detection data collected over space and time, given exogenous multi-source covariate information. %Observations consist of historical glider-based detection records, together with satellite-based covariates. %The primary objective is to estimate the probability that a species is present at a given location and time, conditional on multi-source covariate information, while accounting for imperfect detection. 
Conceptually, the goal is to learn the presence probability as in (\ref{eq:hlevel}).
\begin{equation}
   \Pr\!\big(\text{presence at location } \mathbf{s} \text{ and time } t 
\mid \text{exogenous covariates in $\mathcal{X}$ and $\mathcal{Z}$}\big).
\label{eq:hlevel}
\end{equation}
%We first review the ZIB model formulation and its hierarchical structure in Section \ref{zib}. This is followed by Section \ref{DeepZIB}, where we integrate the deep neural network into the ZIB formulation.
%While the methodology is broadly applicable, we define the spatial and temporal resolution of the model as dictated by our application context. Temporally, the latent presence process is modeled at a monthly resolution. Spatially, the model operates on a continuous domain parameterized by longitude and latitude, without imposing a fixed spatial resolution or predefined spatial grid.
%\textcolor{red}{[@Jiaxiang, here is a good place to define the spatial and temporal resolution of the model.]}
Section \ref{zib} formulates the Zero-Inflated Bernoulli (ZIB) model %to explicitly distinguishing true absence from non-detection within passive acoustic monitoring data 
for PAM data, followed by the proposed DeepZIB generalization in Section \ref{DeepZIB}. %We then present a variant of ZIB that is equipped with a random field for modeling residual spatio-temporal dependencies. 
Section \ref{theory} and \ref{estimation} present its theoretical properties, estimation and prediction steps, respectively.
\subsection{The Zero-Inflated Bernoulli (ZIB) model}
\label{zib}

%Let $Y(\mathbf{s},t) \in \{0,1\}$ denote the binary variable that indicates whether a whale is detected at location $\mathbf s \in \mathbb{R}^2$ and time $t$. Specifically, $Y(\mathbf{s},t) = 1$ if a detection occurs, and $Y(\mathbf{s},t)=0$ otherwise. Let $S(\mathbf{s},t) \in \{0,1\}$ be a latent binary variable representing the unobservable presence states of the whale ($S(\mathbf{s},t)=1$ if the whale is present, and $S(\mathbf{s},t)=0$ otherwise). 
Let $Y(\mathbf{s},t) \in \{0,1\}$ denote the binary detection outcome at location 
$\mathbf{s} \in \mathbb{R}^2$ and time $t$: %defined as
\begin{equation}
Y(\mathbf{s},t) :=
\begin{cases}
1, & \text{if a marine mammal is detected at  location $\mathbf{s}$ and time $t$},\\
0, & \text{otherwise}.
\end{cases}
\end{equation}
Similarly, let $S(\mathbf{s},t) \in \{0,1\}$ be a latent binary variable representing the 
(unobservable) presence state of the marine mammal:
\begin{equation}
S(\mathbf{s},t) :=
\begin{cases}
1, & \text{if a marine mammal is present at location } \mathbf{s} \text{ and time } t,\\
0, & \text{otherwise}.
\end{cases}
\end{equation}
These two processes are not equivalent %since absence of detection does not automatically equate whale absence, 
due to imperfect detections. Let $\mathbf{x}(\mathbf{s},t) \in \mathbb{R}^p$ denote the $p$-dimensional covariate vector associated with the presence process (i.e., covariates affecting habitat preferences). Similarly, let $\mathbf{z}(\mathbf{s},t) \in \mathbb{R}^q$ denote the $q$-dimensional covariate vector associated with the detection process (i.e., covariates affecting the likelihood of detection, conditional on presence). Note that $\mathbf{x}(\mathbf{s},t)$ and $\mathbf{z}(\mathbf{s},t)$ can share some covariates. %To ensure model identifiablity, we assume that there exists at least one continuous variable that appears in $\mathbf{x}(\mathbf{s},t)$ but not in $\mathbf{z}(\mathbf{s},t)$ or vice versa. This assumption prevents the effects of whale presence and detection probability from being completely confounded, allowing the latent presence process and the detection process to be distinguished.

Acoustic detections of marine species presence contain a large portion of zeros, which arise from two mechanisms. %Structural zeros occur when no whale is present at the sampling location $\mathbf{s}$ at time $t$, in which case detection is impossible. Non-structural zeros, by contrast, arise when a whale is present but remains undetected due to silent behavior or sensor limitations. 
Structural zeros occur when $S(\mathbf{s},t)=0$, in which case detection is impossible. Non-structural zeros arise when $S(\mathbf{s},t)=1$ but $Y(\mathbf{s},t)=0$, in which case the whale is present but undetected. %due to silent behavior or sensor limitations. %The ZIB model explicitly separates these two mechanisms by introducing the latent presence state $S$. 
This decomposition into presence and detection components provides a principled explanation for the observed excess zeros.

The function $\pi(\mathbf{x}(\mathbf{s},t);\boldsymbol{\Theta})$ defined in (\ref{eq:general_notation1}) denotes the probability that a whale is present at location $\mathbf{s}$ and time $t$, given the covariates $\mathbf{x}(\mathbf{s},t)$. 
%The presence and detection probability functions for the zero-inflated model are defined as (\ref{eq:general_notation1},\ref{eq:general_notation2})
\begin{equation}
\label{eq:general_notation1}
    \pi(\mathbf{x}(\mathbf{s},t);\boldsymbol\Theta) = P(S(\mathbf{s},t)=1|\mathbf{x}(\mathbf{
s},t)).
\end{equation}
Conditional on whale presence, the function $\varphi(\mathbf{z}(\mathbf{s},t);\boldsymbol{\Psi})$ defined in (\ref{eq:general_notation2}) represents the probability of detecting a whale at location $\mathbf{s}$ and time $t$, given the covariates $\mathbf{z}(\mathbf{s},t)$. 
\begin{equation}
\label{eq:general_notation2}
\varphi(\mathbf{z}(\mathbf{s},t);\boldsymbol\Psi) = P(Y(\mathbf{s},t)=1|S(\mathbf{s},t)=1,\mathbf{z}(\mathbf{s},t))
\end{equation}
where $\boldsymbol \Theta$ and $\boldsymbol \Psi$ parametrize the presence and detection components, respectively. %The exact forms of $\boldsymbol \Theta$ and $\boldsymbol \Psi$ depends on the specific model considered. %For notation convenience, we write $\pi(\mathbf{s},t)$ and $\varphi(\mathbf{s},t)$ as shorthands for $\pi(\mathbf{x}(\mathbf{s},t))$ and $\varphi(\mathbf{z}(\mathbf{s},t))$.

For the classical ZIB model, these parameter sets reduce to vectors of regression coefficients, $\boldsymbol \Theta = \boldsymbol{\alpha}, \boldsymbol{\Psi} = \boldsymbol{\beta}$. %associated with the presence and detection components, respectively. 
The presence and detection components are typically specified through generalized linear models with a logit link, as expressed in (\ref{eq:logit1}) and (\ref{eq:logit2}), respectively.
\begin{equation}
\label{eq:logit1}
\pi_{zib}(\mathbf{x}(\mathbf{s},t);\boldsymbol{\Theta}) = \frac{1}{1 + \exp(-\mathbf{x}(\mathbf{s},t)^\top\boldsymbol{\alpha})},
\end{equation}
\begin{equation}
\label{eq:logit2}
\varphi_{zib}(\mathbf{z}(\mathbf{s},t);\boldsymbol{\Psi}) = \frac{1}{1 + \exp(-\mathbf{z}(\mathbf{s},t)^\top\boldsymbol{\beta})}.
\end{equation}
Furthermore, both the presence and the detection processes are assumed to follow the Bernoulli distribution, as expressed in (\ref{eq:B1}) and (\ref{eq:B2}), respectively.
\begin{equation}
S(\mathbf{s},t) \sim \mathrm{Bernoulli} (\pi(\mathbf{x}(\mathbf{s},t);\boldsymbol{\Theta}))
\label{eq:B1}
\end{equation}
\begin{equation}
Y(\mathbf{s},t)\mid S(\mathbf{s},t) \sim \mathrm{Bernoulli}(\varphi(\mathbf{z}(\mathbf{s},t);\boldsymbol{\Psi})S(\mathbf{s},t)).
\label{eq:B2}
\end{equation}
Hence, the ZIB model can be written as follows:
\begin{align}
    \Pr\bigl(Y(\mathbf{s},t)\! =\! y\!\mid\!  \mathbf{x}(\mathbf{s},t),\mathbf{z}(\mathbf{s},t)\bigr)\!
  =&\! \bigl[1\!-\!\pi(\mathbf{x}(\mathbf{s},t);\boldsymbol{\Theta})\bigr]\mathbf{1}_{\{y=0\}}
   \! + \nonumber\\&\! \pi(\mathbf{x}(\mathbf{s},t);\boldsymbol{\Theta})\varphi(\mathbf{z}(\mathbf{s},t);\boldsymbol{\Psi})^{y}\bigl[1\!-\!\varphi(\mathbf{z}(\mathbf{s},t);\boldsymbol{\Psi})\bigr]^{1-y},   
\end{align}
where observations are assumed to arise from a two-stage process. With a probability of $1-\pi(\mathbf{x}(\mathbf{s},t);\boldsymbol{\Theta})$, the species is absent and the detection is deterministically zero. Meanwhile, with a probability of $\pi(\mathbf{x}(\mathbf{s},t);\boldsymbol{\Theta})$, the species is present and the detection follows a Bernoulli distribution with a detection probability of $\varphi(\mathbf{z}(\mathbf{s},t);\boldsymbol{\Psi})$.

%\textcolor{red}{[Jiaxiang, please add a sentence or two to explain equation 10 here.]}

For $n$ glider samples $(y_i, \mathbf{x}_i, \mathbf{z}_i)_{i=1}^n$, where $y_i = Y(\mathbf{s}_i,t_i), \mathbf{x}_i = \mathbf{x}(\mathbf{s}_i,t_i)$, and $\mathbf{z}_i = \mathbf{z}(\mathbf{s}_i,t_i)$, and by letting $\pi_i=\pi(\mathbf{x}(\mathbf{s}_i,t_i);\boldsymbol{\Theta})$ and $\varphi_i=\varphi(\mathbf{z}(\mathbf{s}_i,t_i);\boldsymbol{\Psi})$, we can express the negative log-likelihood function for the ZIB model as follows:
\begin{equation}
\mathcal{L} = -\sum_{i=1}^n \bigl[y_i log(\pi_i \varphi_i) + (1-y_i) log(1-\pi_i+\pi_i(1-\varphi_i)) \bigr].
\end{equation}

The model parameters can be directly estimated by optimizing $\mathcal{L}$. From there, the ZIB model can be used to obtain predictions of presence at any location $\mathbf{s}$ and time $t$ for which covariate information, $\mathbf{x}(\mathbf{s},t)$, is available. %Specifically, satellite-derived products and glider-based measurements enter the model exclusively through the covariate vectors $\mathbf{x}(\mathbf{s},t)$ and $\mathbf{z}(\mathbf{s},t)$.
Importantly, prediction does not require a new glider deployment and can be generated at locations and times where satellite products are available, enabling broad spatial and temporal coverage beyond the glider surveys. %Despite these advantages, the classical ZIB formulation remains limited in several aspects. Both the presence and detection components are specified using rigid parametric regression models, which restrict their ability to exploit the high-dimensional, multi-modal marine ecosystem data. Moreover, static or weakly dynamic covariate representations limit predictive performance under rapidly evolving environmental conditions. 

\subsection{The proposed deep zero-inflated Bernoulli (DeepZIB) model}
\label{DeepZIB}
\begin{figure}[h]
\centering
\includegraphics[trim=2cm 0cm 2cm 0cm,clip,width = 0.9\linewidth]{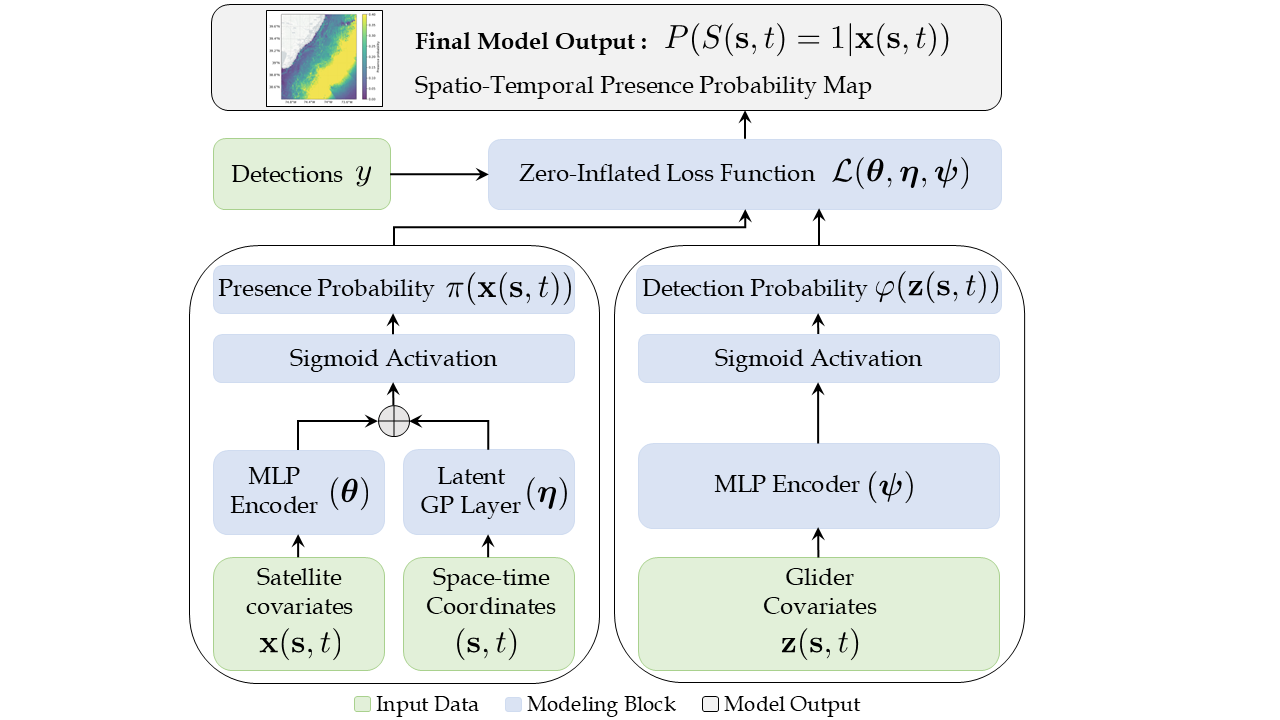}
\caption{%Illustration of the DeepZIB framework. 
The DeepZIB framework. %The architecture consists of two parallel streams modeling 
Two parallel streams model the latent ecological and observation processes. The presence branch (left) integrates satellite-derived environmental covariates $\mathbf x(\mathbf{s},t)$ and spatiotemporal coordinates $(\mathbf s,t)$ to estimate the latent presence probability $\pi(\mathbf{x}(\mathbf{s},t))$. %Simultaneously, 
The detection branch (right) maps glider-based covariates $\mathbf z(\mathbf{s},t)$ to the conditional detection probability $\varphi(\mathbf{z}(\mathbf{s},t))$. These %probabilities 
are combined with the observed detections %labels 
$y$ %to optimize the zero-inflated loss function, enabling joint optimization of network weights and GRF parameters. 
for joint parameter estimation, and subsequently prediction at new inputs. %The fitted parameters are then
%The estimate model can then be applied to new inputs %to generate predictions of whale presence probability.
%for predicting whale presence.
}
\label{fig:frame}
\end{figure}
%In the classical ZIB model, both the presence and detection probabilities are modeled through logit–linear predictors of the corresponding covariates. While interpretable, these linear specifications can be restrictive in applications where covariate effects are nonlinear or interact in complex ways. 
In the classical ZIB model, both the presence and detection functions are specified using parametric regression models. %with static or weakly dynamic %which restrict their ability to exploit the high-dimensional, multi-modal marine ecosystem data. Moreover, static or weakly dynamic covariate representations.%, limiting predictive performance under rapidly evolving environmental conditions. 
%The limitations mentioned above motivate the development of a DeepZIB framework, in which flexible representation learning is used to model the presence and detection process.
Instead, DeepZIB models these functions 
through internalized deep neural networks (DNNs), as %The high-level workflow of DeepZIB is shown 
illustrated in Figure \ref{fig:frame} and explained below.

Let $\boldsymbol{\theta}$ and $\boldsymbol{\psi}$ denote the network weights for the presence and detection functions, respectively, such that $\boldsymbol{\Theta} = \boldsymbol{\theta}$ and $\boldsymbol{\Psi} = \boldsymbol{\psi}$. The corresponding DeepZIB functions are denoted by $\pi_{deep}(\mathbf{x}(\mathbf{s},t);\boldsymbol{\Theta})$ and $\varphi_{deep}(\mathbf{z}(\mathbf{s},t);\boldsymbol{\Psi})$, and can be expressed as in (\ref{eq:DeepZIB_prob1}) and (\ref{eq:DeepZIB_prob2}), respectively. 
\begin{equation}
\label{eq:DeepZIB_prob1}   \pi_{deep}(\mathbf{x}(\mathbf{s},t);\boldsymbol{\Theta}) = \frac{1}{1+\exp (-f_{\boldsymbol{\theta}}(\mathbf{x}(\mathbf{s},t)))},
\end{equation}
\begin{equation}
\label{eq:DeepZIB_prob2} 
\varphi_{deep}(\mathbf{z}(\mathbf{s},t);\boldsymbol{\Psi}) = \frac{1}{1+\exp (-h_{\boldsymbol{\psi}} (\mathbf{z}(\mathbf{s},t)))},
\end{equation}
where $f_{\boldsymbol{\theta}}$ and $h_{\boldsymbol{\psi}}$ are distinct DNN models.  
%The DeepZIB model is a generalization of the ZIB model. Specifically, DeepZIB with a shallow multi-layer perceptron (no hidden layers) and a sigmoid activation function is indeed a classical ZIB model. This equivalence will be formally established in Section \ref{theory}. 
The choice of the network architecture for $f_{\boldsymbol{\theta}}$ and $h_{\boldsymbol{\psi}}$ is flexible, so long that $f_{\boldsymbol{\theta}}$ and $h_{\boldsymbol{\psi}}$ remain differentiable with respect to their corresponding parameter sets, $\theta$ and $\psi$, respectively. This condition ensures the differentiability of the loss function for gradient-based optimization.
%For the DNN components of the proposed model, 
Here, we adopted a multilayer perceptron (MLP) to model each of $f_{\boldsymbol{\theta}}$ and $h_{\boldsymbol{\psi}}$. Specifically, the input covariates are first normalized, then ReLU activations are used throughout the hidden layers, followed by a sigmoid transformation at the output layer to ensure valid probability values. The number of layers and neurons for each network can be different, as discussed in Section \ref{estimation}. 

We then formulate a %EM-inspired 
cross-entropy loss function that generalizes that of the classical ZIB:
\begin{align}
\label{eq:loss}
%\small
  \mathcal{L}^{(deep)} = & \sum_{i=1}^{n} \underbrace{-\,y_i \,\log(\pi_i)
  \;-\;
  (1 - y_i)\,
  \Bigl[
      w_i \,\log(\pi_i)
      + (1 - w_i)\,\log\bigl(1-\pi_i\bigr)
  \Bigr]}_{\textit{Non-Structural Zeros}} \\
  \nonumber &-\underbrace{\,y_i \,\log(\varphi_i)
  \;-\;
  (1 - y_i)\,w_i \,\log\bigl(1-\varphi_i\bigr).}_{\textit{Detections within the Non-Structural Subset}}
% \mathcal{L} = \sum_{i=1}^n (l_{1,i} + l_{2,i})
\end{align}
\iffalse
\begin{align}
\mathcal{L}_{1}
&=
\sum_{i=1}^{n}
\big[
- y_i \log(\pi_i)
- (1-y_i)( w_i \log(\pi_i) + (1-w_i)\log(1-\pi_i))
\big], \\ \nonumber
\mathcal{L}_{2}
&=
\sum_{i=1}^{n}
\big[
- y_i\log(\varphi_i)
- (1-y_i) w_i\log(1-\varphi_i)
\big]. \\ \nonumber
\mathcal{L} &= \mathcal{L}_1 + \mathcal{L}_2
\end{align}
\fi
The loss in \eqref{eq:loss} comprises two components that correspond to distinct parts of the data-generating process. 
Importantly, we explicitly leverage the fact that a detection constitutes a confirmed presence. That is, for observations with $y_i = 1$, the latent state is known to be non-structural ($s_i=1$), whereas uncertainty about the latent state arises only for non-detections ($y_i=0$). Hence, the first component in (\ref{eq:loss}) evaluates the prediction of non-structural zeros. For observations with $y_i=0$, the latent indicator distinguishing structural from non-structural zeros is unobserved. We therefore introduce $w_i = \Pr (\mathbf{s}_i=1|y_i =0)$ as the posterior probability that a non-detection arises from the non-structural state. This quantity acts as a soft label that allows each zero to contribute probabilistically to the first component of the loss. The second component in (\ref{eq:loss}) evaluates detection within the non-structural subset. Conditional on being non-structural, detections are governed by the detection probability $\varphi_i$. This loss penalizes mismatches between observed detections and %predicted 
detection probabilities, with $w_i$ %to 
reflecting uncertainty about whether detection was possible. 

%\textcolor{red}{[@Jiaxiang, do we here leverage the fact that a detection is a confirmed presence? If yes, we need to explicitly mention it.]}

%Here, the first part evaluates the prediction of non-structural zeros, whereas the second part evaluates the detections within the non-structural subset. $w_i = P(s_i = 1|y_i = 0)$ denotes the posterior probability that an observation with $y_i=0$ is a non-structural zero. $w_i$ serves as a soft label in the first term (since $s_i$ is unobservable) and as a weight in the second term. 

%At each training iteration, we sequentially optimize the loss function, first run a forward pass with the current network parameters to obtain $\pi_i,\varphi_i$, and then compute the posterior weights $w_i$. After that, we evaluate the two cross-entropy terms with fixed weights and back-propagate to update the network parameters that generate $\pi_i, \varphi_i$.  

%\subsection{DeepZIB-ST model}
%\label{DeepZIBst}

Beyond the covariate effects in the presence function captured through $f_{\boldsymbol{\theta}}$, spatial and temporal dependence can be modeled by augmenting the presence predictor in the DeepZIB model with a spatio-temporal Gaussian random field (GRF). In this case, the parameter set is augmented to include both the network weights $\boldsymbol{\theta}$ and the GRF hyperparameters $\boldsymbol{\eta}$, that is, we have $\boldsymbol{\Theta}=(\boldsymbol{\theta},\boldsymbol{\eta})$. For an observation at location $\mathbf{s}$ and time $t$, the presence probability, now denoted as $\pi_{deepst}(\mathbf{x}(\mathbf{s},t);\boldsymbol{\Theta}),$ can be expressed as in (\ref{eq:deepst}). 
\begin{equation}
\pi_{deepst}(\mathbf{x}(\mathbf{s},t);\boldsymbol{\Theta}) = \frac{1}{1+\exp(-(f_{\boldsymbol{\theta}}(\mathbf{x}(\mathbf{s},t))+g_{\boldsymbol{\eta}}(\mathbf{s},t)))},
\label{eq:deepst}
\end{equation}
where $g_{\boldsymbol{\eta}}(\mathbf{s},t)$ is a zero-mean GRF with a stationary covariance function denoted by $K(\mathbf{w},u)$, such that $\mathbf{w}$ and $u$ are the spatial and temporal lags, respectively. 
%\begin{equation}
%    g_{\boldsymbol{\eta}}(\mathbf{s},t) \sim \mathcal{GP}(0, K(\mathbf{w},u)), 
%\end{equation}
%Here $K(\mathbf{w},u)$ is a stationary kernel parameterized by $\boldsymbol{\eta}$, where $\mathbf{w},u$ are the spatial and temporal lags. The most prevalent approach to specify $K(\mathbf{w},u)$ in the literature is through the separable approach, which decomposes the dependence structure over space and time such that 
Here, $K(\mathbf{w},u)$ can be modeled by decomposing the dependence structure over space and time, as in (\ref{eq:kernel}). 
\begin{equation}
\label{eq:kernel}
    K(\mathbf{w},u) = \kappa \bigg(  K^s(\mathbf{w})\times K^t(u)\bigg),
\end{equation}
such that $\kappa >0$ is the marginal variance, whereas $K^s(\mathbf{w})$ is a squared exponential kernel as in (\ref{eq:kernel1}), and $K^t(u)$ is a periodic kernel, as expressed in (\ref{eq:kernel2}), encoding seasonal periodicity. %with a period = 12. 
%For the spatial component, we use the squared exponential kernel (\ref{eq:kernel1}), while the temporal component is modeled using a periodic kernel defined in (\ref{eq:kernel2}):
\begin{equation}
\label{eq:kernel1}
    K^s(\mathbf{w}) = \exp\bigg(-\frac{\|\mathbf{w}\|^2}{2r_w^2}\bigg),
\end{equation}
\begin{equation}
\label{eq:kernel2}
K^t(u) = \exp\bigg(-\frac{2\sin^2(\pi |u|/12)}{r_u^2}\bigg),
\end{equation}
where $r_w$ and $r_u$ are spatial and temporal length-scale parameters, respectively.

%Meanwhile, the detection model retains its own parameter vector $\boldsymbol{\Psi}$, as defined in Section \ref{DeepZIB}. This choice reflects the different roles of the presence and detection models in DeepZIB. The presence process represents the underlying ecological habitat relationships, %occurrence, 
%for which smooth spatial and temporal variation is expected, and in fact, desirable. 
%while 
Since the detection process %is conditional on presence and is primarily 
is driven by survey conditions, %Since this component does not aim to represent large-scale ecological structure, and introducing a second latent random field would increase model complexity, 
we do not include an additional GRF term in the detection model. We refer to this GRF-augmented model as DeepZIB-ST.  

\subsection{Theoretical properties}
\label{theory} 
%Here, we establish important theoretical properties of the proposed framework. %followed by key properties of the loss function proposed in (\ref{eq:loss}) and the model structure. 
For %the %below theoretical properties, 
%theoretical analysis in 
notational simplicity in Section \ref{theory}, we omit the space-time coordinates $\mathbf{s}$ and $t$, %for notational simplicity. We also omit 
as well as the subscripts $deep$ or $deepst$. %since our focus is mainly on DeepZIB and DeepZIB-ST models.
\subsubsection{Parameter Identifiability in DeepZIB} 
\iffalse
%First, we formalize the identifiability problem, and then we show that minimizing the proposed loss function is equivalent to minimizing the negative log-likelihood. Finally, we prove that the DeepZIB model reduces to the classical ZIB model when the neural networks degenerate to a single linear layer, thereby justifying the DeepZIB as a strict generalization of the classical formulation.
A model is said to be identifiable when different parameter values generate different distributions for the observable data. %In other words, if two parameter sets produce exactly the same distribution of the observed data, then those parameter sets must encode the same model. 
For the DeepZIB model, this means that, whenever two parameter sets %$(\boldsymbol{\Theta},\boldsymbol{\Psi})$ and $(\boldsymbol{\Theta}',\boldsymbol{\Psi}')$ 
yield the same distribution for the observable detection process across all covariate values, then they must also yield the same underlying presence and detection functions. This is formally defined in Definition
\fi
Definition \ref{def:ident_zib} delineates identifiability conditions for DeepZIB, where $\mathcal{X}$ and $\mathcal{Z}$ denote the feature spaces of the presence and detection covariate vectors $\mathbf{x}(\mathbf{s},t)$ and $\mathbf{z}(\mathbf{s},t)$, respectively.   
\begin{definition}\label{def:ident_zib}
The DeepZIB model is said to be identifiable if, for any two parameter sets
$\boldsymbol \vartheta = (\boldsymbol \Theta, \boldsymbol \Psi)$ and
$\boldsymbol \vartheta^* = (\boldsymbol \Theta^*, \boldsymbol \Psi^*)$,
$\Pr(Y=y \mid \mathbf{x}, \mathbf{z}; \boldsymbol \vartheta)
=
\Pr(Y=y \mid \mathbf{x}, \mathbf{z}; \boldsymbol \vartheta^*)$
$
\text{for all } (\mathbf{x},\mathbf{z}) \in \mathcal X \times \mathcal Z
$
implies that
$
\pi(\mathbf{x};\boldsymbol \Theta) = \pi(\mathbf{x};\boldsymbol \Theta^*)$
$ \text{for all } \mathbf{x}\in\mathcal X,
\varphi(\mathbf{z};\boldsymbol \Psi) = \varphi(\mathbf{z};\boldsymbol \Psi^*)
$
$ \text{for all } \mathbf{z}\in\mathcal Z.
$
\end{definition}

To establish identifiability of the DeepZIB model, we require a mild structural assumption that the set of detection covariates $\mathcal Z$ contains at least one continuous variable that is not included in the presence covariate set $\mathcal X$. This assumption is well aligned with the distinct roles of the presence and detection components, as covariates governing detection are typically different from those driving ecological presence.
Under this assumption, Lemma \ref{lm:id} establishes identifiability by exploiting the factorization of the conditional probability of observation into a presence component and a detection component.
%To show the identifiability of the DeepZIB model, we only make an assumption that $\mathcal{Z}$ contains at least one continuous variable that is not included in $\mathcal{X}$. This is a practically non-problematic assumption given the different roles of the detection and presence functions. Lemma \ref{lm:id} establishes identifiability by using the fact that the observed detection probability can be factorized into two components, and assumes at least one continuous covariate appears in only one of the two components.

%From there, it is clear that if $\pi(\mathbf{x};\boldsymbol{\Theta}) = \pi(\mathbf{x};\boldsymbol{\Theta^*})$ $\forall \mathbf{x}\in \mathcal{X}$, $\varphi(\mathbf{z};\boldsymbol \Psi) = \varphi(\mathbf{z};\boldsymbol \Psi^*)$ $\forall \mathbf{z} \in \mathcal{Z}$, then we would have $Pr(Y=y|\mathbf{x},\mathbf{z};\boldsymbol \vartheta) = Pr(Y=y|\mathbf{x},\mathbf{z}; \boldsymbol \vartheta^*)$. From there, we only need to focus on showing that if $Pr(Y=y|\mathbf{x},\mathbf{z};\boldsymbol \vartheta) = Pr(Y=y|\mathbf{x},\mathbf{z}; \boldsymbol \vartheta^*)$, then $\pi(\mathbf{x};\boldsymbol{\Theta}) = \pi(\mathbf{x};\boldsymbol{\Theta^*})$ $\forall \mathbf{x}\in \mathcal{X}$, $\varphi(\mathbf{z};\boldsymbol \Psi) = \varphi(\mathbf{z};\boldsymbol \Psi^*)$ $\forall \mathbf{z} \in \mathcal{Z}$.

%\textcolor{red}{[@Jiaxiang, this last statement is unclear.]}

\begin{lemma}
\label{lm:id}
Under the assumption that there exists a continuous variable $z_\ell\in\mathcal Z\setminus\mathcal X$
such that $\partial_{z_\ell}\log\varphi(\mathbf{z};\boldsymbol \psi)$ exists,
the DeepZIB model is identifiable if $Pr(Y=y|\mathbf{x},\mathbf{z};\boldsymbol \vartheta) = Pr(Y=y|\mathbf{x},\mathbf{z};\boldsymbol \vartheta^*)$ implies $\pi(\mathbf{x};\boldsymbol{\Theta}) = \pi(\mathbf{x};\boldsymbol{\Theta^*})$ $\forall \mathbf{x}\in \mathcal{X}$, $\varphi(\mathbf{z};\boldsymbol \Psi) = \varphi(\mathbf{z};\boldsymbol \Psi^*)$ $\forall \mathbf{z} \in \mathcal{Z}$
\end{lemma}

Lemma \ref{lm:id} establishes that, under the stated conditions, the presence and detection components of DeepZIB are identifiable from the observed data. This result provides the theoretical basis for distinguishing latent absence from imperfect
detection. %in the proposed framework. 
The proof is in the appended Supplementary Materials (Section \ref{sec:proof}).

%\textcolor{red}{[@Jiaxiang, the below is unclear]}
%To facilitate stable optimization of the proposed zero-inflated model, we consider an alternative objective function $\mathcal{L}$ that is equivalent to the negative log-likelihood in terms of optimization. The following lemma establishes this equivalence by showing that both objectives induce identical gradients.
%\textcolor{red}{[@Jiaxiang, need to first introduce what lemma 3 is or what do we try to show.]}
\subsubsection{Connection to ZIB}
Lemma \ref{lemma3} shows that DeepZIB reduces to the classical ZIB model under linear parameterization, thereby establishing DeepZIB as a generalization of the classical ZIB formulation.
\begin{lemma}
\label{lemma3}
Consider the DeepZIB model with
$\pi(\mathbf{x};\boldsymbol{\Theta})
  = \frac{1}{1+\exp(-f_{\boldsymbol{\theta}}(\mathbf{x}))}, 
\varphi(\mathbf{z};\boldsymbol{\Psi})
  = \frac{1}{1+\exp(-h_{\boldsymbol{\psi}}(\mathbf{z}))}$, 
where $f_{\boldsymbol{\theta}}$ and $h_{\boldsymbol{\psi}}$ are feed-forward neural networks. If each network consists of a single affine layer with no hidden layers and no intermediate nonlinear activations,
\iffalse
that is,
$f_{\boldsymbol{\theta}}(\mathbf{x})
  = \mathbf{x}^\top\boldsymbol{\alpha},
h_{\boldsymbol{\psi}}(\mathbf{z})
  = \mathbf{z}^\top\boldsymbol{\beta}
$, for some coefficient vectors $\boldsymbol{\alpha}$ and $\boldsymbol{\beta}$.
Then, for any $(\mathbf{x},\mathbf{z})$,
\fi
then the DeepZIB model reduces to the classical ZIB model.
\iffalse
with logit–linear predictors
$\pi(\mathbf{x};\boldsymbol{\alpha})
  = \frac{1}{1+\exp(-\mathbf{x}^\top\boldsymbol{\alpha})},
\varphi(\mathbf{z};\boldsymbol{\beta})
  = \frac{1}{1+\exp(-\mathbf{z}^\top\boldsymbol{\beta})}.$
Conversely, any classical ZIB model with logit–linear predictors can be
represented as a DeepZIB model whose networks have a single linear layer.
\fi
\end{lemma}

\begin{proof}
If $f_{\boldsymbol{\theta}}$ and $h_{\boldsymbol{\psi}}$ each consist of a
single affine map, then there exist coefficient vectors $\boldsymbol \alpha$ and $\boldsymbol \beta$ such that 
$f_{\boldsymbol{\theta}}(\mathbf{x})=\mathbf{x}^\top\boldsymbol{\alpha}$ and
$h_{\boldsymbol{\psi}}(\mathbf{z})=\mathbf{z}^\top\boldsymbol{\beta}$. Substituting into
$\pi(\mathbf{x};\boldsymbol{\Theta})$ and $\varphi(\mathbf{z};\boldsymbol{\Psi})$
yields exactly the same form of the classical ZIB model.
\iffalse
Conversely, given any classical ZIB model with coefficients
$(\boldsymbol{\alpha},\boldsymbol{\beta})$, choosing
$f_{\boldsymbol{\theta}}(\mathbf{x})=\mathbf{x}^\top\boldsymbol{\alpha}$ and
$h_{\boldsymbol{\psi}}(\mathbf{z})=\mathbf{z}^\top\boldsymbol{\beta}$
defines one-layer networks that reproduce the same predictors. Thus the
classical ZIB model is precisely the subclass of DeepZIB models whose networks
contain a single affine layer.
\fi
\end{proof}

%\textcolor{red}{@Jiaxiang,no need for the below paragraph. Please distribute the below text on the individual components above}
 %Lemma \ref{lemma3} shows that DeepZIB reduces to the classical ZIB model under linear parameterization, thereby establishing DeepZIB as a strict generalization of the classical formulation. Together, these results ensure that the presence and detection components are uniquely determined by the observable data. Moreover, the gradient-equivalence result confirms that the proposed objective function is consistent with likelihood-based optimization. Collectively, these properties establish DeepZIB as a well-defined, identifiable extension of the classical ZIB model.

\subsection{Estimation and prediction using DeepZIB}
\label{estimation}
%The estimation of the DeepZIB model follows from the loss function introduced in Section \ref{DeepZIB}. 
The estimation and prediction steps for DeepZIB are detailed in Algorithm \ref{alg:DeepZIB}. Given the current estimates for the parameter sets $(\boldsymbol{\Theta}, \boldsymbol{\Psi})$, the model first runs the forward pass with current network parameters to determine $\pi$ and $\varphi$, and then computes the posterior weight $w$. Fixing $w$, we update the loss function $\mathcal L^{(deep)}$, and then back-propagate
to update the network parameters $(\boldsymbol{\Theta}, \boldsymbol{\Psi})$. For a new space-time observation $(\mathbf{x}^*, \mathbf{z}^*, \mathbf{s}^*, t^*)$, the trained model evaluates the presence probability 
$\hat\pi$
and the conditional detection probability 
$\hat\varphi$.

\begin{algorithm}[htbp]
\caption{Estimation and prediction for the DeepZIB model}
\label{alg:DeepZIB}
\raggedright

\KwIn{Training samples $\{(\mathbf{x}_i,\mathbf{z}_i,\mathbf{s}_i,t_i,y_i)\}_{i=1}^n$; 
Learning rate $\rho$; Number of epochs $T$; Batch size $B$; Number of bootstrap replicates $K$; Coverage level $1-\alpha$; Prediction samples $\{(\tilde{\mathbf{x}}_j,\tilde{\mathbf{z}}_j,\tilde{\mathbf{s}}_j,\tilde{t}_j)\}_{j=1}^m$, constructed at spatial and temporal resolutions $R_s$ and $R_t$}

\textbf{Initialize} network parameters $(\boldsymbol{\theta},\boldsymbol{\eta},\boldsymbol{\psi})$.

\For{$t = 1,\dots,T$}{%
  partition training indices into mini-batches $\mathcal{B}$ of size $B$\;
  \For{each mini-batch $\mathcal{B}$}{%
    \textbf{Forward step:} For each $i\in\mathcal{B}$, compute $\pi_i$ and $\varphi_i$ according to Equation (\ref{eq:deepst}) and (\ref{eq:DeepZIB_prob2}), and 
    $
    w_i = \frac{\pi_i(1-\varphi_i)}{1-\pi_i + \pi_i(1-\varphi_i)}
    $
    
    \textbf{Mini-batch objective:}
    $\displaystyle
\begin{aligned}
\mathcal{L}_{\mathcal{B}}^{(\mathrm{deep})}
= -\frac{1}{|\mathcal{B}|}\sum_{i\in\mathcal{B}}
\Bigl[
&\bigl(y_i+(1-y_i)w_i\bigr)\log\pi_i+(1-y_i)(1-w_i)\log(1-\pi_i)\\
&+y_i\log\varphi_i
 +(1-y_i)w_i\log(1-\varphi_i)
\Bigr].
\end{aligned}
$

    \textbf{Parameter update:} All parameters are updated jointly using a gradient-based optimizer (Adam)
    \[(\boldsymbol{\theta},\boldsymbol{\eta},\boldsymbol{\psi})
    \leftarrow 
    \text{Adam}\bigl((\boldsymbol{\theta},\boldsymbol{\eta},\boldsymbol{\psi}),
                     \nabla\mathcal{L}_{\mathcal{B}}^{(deep)}, \rho\bigr).
    \]
  }
}
Using the fitted parameters $(\boldsymbol{\theta}^*,\boldsymbol{\eta}^*,\boldsymbol{\psi}^*)$, compute $
\hat{\pi}_j$ according to Equation (\ref{eq:deepst})

\For{$k=1,\dots,K$}{%

Generate a bootstrap training dataset
$\mathcal{D}^{*(k)}$.

Refit the model to obtain
$(\boldsymbol{\theta}^{*(k)},
\boldsymbol{\eta}^{*(k)},
\boldsymbol{\psi}^{*(k)})$;

Compute
$\hat{\pi}_j^{(k)}$
for each prediction sample according to Equation (\ref{eq:deepst});
}
For each prediction sample, construct the uncertainty interval using Equation (\ref{eq:uq_interval}).

\KwOut{Estimated parameters $(\boldsymbol{\theta}^*,\boldsymbol{\eta}^*,\boldsymbol{\psi}^*)$, Presence probabilities $\{\hat{\pi}_j\}_{j=1}^m$, Uncertainty intervals $\{\mathcal{I}_{1-\alpha,j}\}_{j=1}^m$}

\end{algorithm}
%The proposed DeepZIB models were implemented using PyTorch and GPyTorch. 
For the embedded DNNs, we find that it is logical for the presence network $f_{\boldsymbol{\theta}}$ to have a more complex architecture than the detection network $h_{\boldsymbol{\psi}}$. Specifically, the MLP used to model the presence probability consists of four hidden layers with widths of 128, 64, 32, and 16, respectively, whereas the detection MLP contains three hidden layers with widths of 64, 32, and 16, respectively. All hidden layers use ReLU activations and are followed by a final linear layer with a sigmoid activation. %To explicitly model the spatiotemporal dependencies in whale presence, we integrated a GRF into the network architecture. Direct inference with a full GP is computationally difficult for large datasets, therefore we 
For the GRF model, we employed a variational Gaussian Process approximation %using a set of inducing points, enabling 
for scalable model training. %The entire framework was trained by minimizing the proposed loss function. 
The optimization was performed using the Adam optimizer, whereas a grid search was used to tune the number of hidden layers, the number of epochs, and the learning rate. 

When generating spatio-temporal probability maps, the spatial and temporal resolution of the prediction outputs can be chosen in light of the downstream spatial management task, subject to the availability and resolution of the input covariates. As such, $R_s$ and $R_t$ can be specified by the user \textit{ex ante} (e.g., monthly time intervals and kilometer-scale spatial resolution). %denote the spatial and temporal resolutions, respectively, at which the prediction maps are constructed. These resolutions can be adapted to the spatial and temporal scales relevant to a particular ecological or operational application, subject to 
Given $R_s$ and $R_t$, the corresponding covariates and spatial coordinates are constructed and passed through the trained model to obtain $\hat\pi(\mathbf{s}, t)$.
\iffalse
To generate spatial probability maps, monthly satellite covariate fields are first constructed by 
temporally averaging the original satellite products within each month. 
These monthly covariate maps have a regular spatial grid 
with a resolution of $0.01^\circ \times 0.01^\circ$ ($\approx 1.24$ km$^2$). For each grid point, the corresponding covariates and spatial coordinates 
are passed through the trained model to obtain $\hat\pi(\mathbf{s}, t)$, 
which are then visualized as monthly presence probability maps.
\fi

%In addition to point estimates of the latent presence probability, we quantify the associated estimation uncertainty using a nonparametric bootstrap. 
For uncertainty quantification (UQ), we use a non-parametric block bootstrap approach \citep{lahiri2003resampling} that is aligned with the sequential sampling structure of the glider data. Let $\mathcal{D} = \{(\mathbf{x}_i,\mathbf{z}_i, \mathbf{s}_i, t_i, y_i)\}_{i=1}^n$ denote the observed training data. %Since observations collected within short time windows are also geographically close along the trajectory, they may exhibit temporal and spatial dependence. To retain this local dependence during resampling, 
Specifically, we partition the training data into contiguous 24-hour blocks and generate each bootstrap dataset by sampling these blocks with replacement. All observations within a selected block are retained together, preserving the dependence structure within that segment of the glider trajectory. %We select a 24-hour block length to provide a balance between preserving local dependence and retaining a sufficient number of resampling units. 
We generate $K$ bootstrap datasets $\{\mathcal{D}^{(k)}\}_{k=1}^K$, by sampling the training data with replacement. For each sample, the DeepZIB model is refitted using the same model architecture and configuration, yielding parameter estimates $\boldsymbol{\Theta}^{*(k)}$ and $\boldsymbol{\Psi}^{*(k)}$. These are then used to compute a bootstrap estimate of the latent presence probability at %a prediction point 
$(\mathbf{s},t)$, %are given by: %expressed in (\ref{eq:uq_pi}).
\iffalse
\begin{equation}
\label{eq:uq_pi}
    \hat{\pi}^{(k)}(\mathbf{s},t) = \pi_{deepst}(\mathbf{x}(\mathbf{s},t);\boldsymbol{\Theta}^{*(k)}), \quad k= 1,\cdots, K.
\end{equation}
\fi 
which is denoted as $\hat{\pi}^{(k)}(\mathbf{s},t)$. The empirical distribution $\{\hat{\pi}^{(k)}(\mathbf{s},t)\}_{k=1}^K$ characterizes the sampling variability of the estimated latent presence probability. This uncertainty is conditional on the covariate vector $\mathbf{x}(s,t)$, assuming that the covariate values at the prediction point are known. Accordingly, a $100(1-\alpha)\%$ equal-tailed uncertainty interval is constructed as (\ref{eq:uq_interval}).
\begin{equation}
\label{eq:uq_interval}
    \mathcal{I}_{1-\alpha}(\mathbf{s},t) = [Q_{\alpha/2}(\{\hat{\pi}^{(k)}(\mathbf{s},t)\}_{k=1}^K), Q_{1-\alpha/2}(\{\hat{\pi}^{(k)}(\mathbf{s},t)\}_{k=1}^K)],
\end{equation}
where $Q_q(\cdot)$ denotes the empirical $q$-th quantile of the bootstrap distribution.

\section{Results}
\label{sec:results}
%This section presents the experimental evaluation of the proposed models. First, 
We first present our experimental setup, followed by simulation studies in Section \ref{sec: sim}, and the application to NARW presence modeling in Section \ref{sec: case}.

\subsection{Experimental Setup} \label{subsec: exp}
We compare DeepZIB against %a wide array of benchmark models that represent the primary alternative modeling strategies for this type of data. These benchmarks span both machine learning and statistical approaches, and differ in how they treat the observation process and zero inflation. 
%\begin{enumerate}
%    \item 
%(\textit{i}) 
(1) \textit{Statistical and ML models}, including Logistic Regression (LR), Extreme Gradient Boosting (XGBoost), and Multilayer Perceptrons (MLP), which learn a direct mapping from covariates to observed detections, but assume %and are widely used in species distribution modeling. However, 
perfect detection (no latent presence state); and %thereby lacking a mechanism to differentiate between true ecological absence and non detections. %As a result, all zero observations are treated equivalently, which can lead to biased inference when detection is imperfect.%(\textit{ii}) 
%\item 
(2) \textit{Classical ZIB model,} which separates structural zeros from non-structural zeros to account for excess zeros, as %This is the same model 
described in Section \ref{zib}. 
%\end{enumerate}

%However, classical ZIB models typically rely on rigid parametric link functions, which limit their ability to capture complex nonlinear relationships in multi-source, multi-modal data settings.
%We evaluate the performance of our proposed models (DeepZIB, DeepZIB-ST) against a wide array of benchmark models. These benchmarks were chosen to represent the primary alternative modeling strategies for this type of data. (\textit{i}) machine learning models including Logistic Regression (LR), eXtreme Gradient Boosting (XGBoost), Multilayer Perceptrons (MLP): This group of models represents a common approach that ignores the observation process (i.e., assumes detection is perfect). (\textit{ii}) classical ZIB model: This model is the traditional method that accounts for the zero-inflated nature of the data. 

%The presence probability is the primary target of interest, hence 
Evaluations are performed on the predicted presence probabilities $\hat{\pi}_i$ %For the simulation case studies (details of which are explained in Section \ref{sec: sim}), the ground truth presence labels are available
%Because the simulation study provides access to both the ground-truth presence labels and the true presence probabilities, whereas the case study contains only binary detection outcomes, we adopt different sets of evaluation metrics in the two settings.
%For real-world case study, models are trained using observations collected prior to February 2022 and evaluated on data collected later. 
%The model performance is assessed 
using the following metrics: the area under the receiver operating characteristic curve (AUC), $F_1$ score, mean absolute error (MAE), Brier score, negative log-likelihood (NLL), and Winkler score (for UQ). %Higher values indicate better performance for AUC and $F_1$ score, whereas lower values are preferred for MAE, Brier score, and NLL. In addition to point estimation accuracy, we use a nonparametric bootstrap to provide an empirical characterization of uncertainty in the estimated latent presence probability. 
%Uncertainty intervals are evaluated using the Winkler score, which jointly characterize interval calibration and sharpness. Formal definitions of 
The evaluation metrics are mathematically defined in the Supplementary Materials (Section \ref{sec:metrics}).
For the simulation experiments (Section \ref{sec: sim}), ground truth labels and presence probabilities are available, so all evaluation metrics can be calculated. In the real-world case study (Section \ref{sec: case}), the true presence probabilities are unknown and only detection labels $y_i$ are observed, so we can only evaluate %the models using 
the Brier score and NLL. 

\subsection{Simulated case studies}
\label{sec: sim}
%We first conduct a simulated case study to assess whether the proposed estimation procedure can recover the true regression coefficients when the data are generated from a zero-inflated Bernoulli (ZIB) model. 
We simulate the detection and presence data along an actual glider path from our real-world dataset, using co-located glider- and satellite-based covariates for model training and satellite covariates beyond the sampled trajectory for spatio-temporal extrapolation.

Next, we consider two simulated cases for generating the true presence and detection data, referred to as Case Studies I and II. In Case Study I, we simulate presence and detection probabilities based on a simple linear functional form, as defined in (\ref{eq:linear1}) and (\ref{eq:linear2}), respectively, such that $x_{ij}$ is the $i$th observation of the $j$th covariate in $\mathbf{x}$. Similarly, $z_{ij}$ is the $i$th observation of the $j$th covariate in $\mathbf{z}$. 
%\textcolor{red}{[Jiaxiang, please check the following and add information if needed.]} 
%Here, the covariate vector for the presence process is defined as: 
The presence covariate vector is defined as $\mathbf{x} = (\text{Frontal Value, Sea Surface Temperature, Chlorophyll, Distance to shore, Seasonal term)}$,
where the distance to shore is calculated based on the Euclidean distance from the points to the coastline, whereas the seasonal term is the sine and cosine transformation of the month. The detection covariate vector %for the detection process 
is defined as: $\mathbf{z} = (\text{Depth, Temperature, Salinity, Oxygen concentration, Elevation, Distance to shore)}.$ 
\begin{equation}
\label{eq:linear1}
\pi_i = \frac{1}{1 + \exp\big(-(-1.5 +3 x_{i1} - x_{i2}
+ x_{i3} +0.8 x_{i4} -0.4 x_{i5} + 0.8 x_{i6})\big)}, 
\end{equation}
\begin{equation}
\label{eq:linear2}
\phi_i = \frac{1}{1 + \exp\big(-(0.1 + 0.4 z_{i1} +0.5z_{i2}
- 0.9 z_{i3} - z_{i4} + 0.3 z_{i5}
+ 0.5 z_{i6})\big)}.
\end{equation}
\iffalse
\begin{align}
\label{eq:linear}
\pi_i
&= \frac{1}{1 + \exp\big(-(-0.5 + 0.01 x_{i1} - x_{i2}
- 0.8 x_{i3} + 0.7 x_{i4} - x_{i5} - 0.4 x_{i6})\big)}, \\ \nonumber
\phi_i
&= \frac{1}{1 + \exp\big(-(0.5 - 0.8 z_{i1} - z_{i2}
- 0.2 z_{i3} - z_{i4} - 0.3 z_{i5}
+ 0.2 z_{i6} - 0.1 z_{i7})\big)}.
\end{align}
\fi 
In Case Study II, presence and detection probabilities are simulated using a more complex, non-linear functional form, as expressed in (\ref{eq:nonlinear1}) and (\ref{eq:nonlinear2}), respectively. 
% \begin{equation}
% \small 
% \label{eq:nonlinear1}
% \begin{aligned}
% \pi_i
% &= \frac{1}{1 + \exp\Big(-\big[-1.2 + 0.6 x_{i1} - 1.5 x_{i2} +1.7 x_{i3}
% + 1.5 x_{i4} -  x_{i5} +0.4 x_{i6} \\
% &\quad
% + 1.5 x_{i1} x_{i2} + 0.2 x_{i4} x_{i5} + 1.6 x_{i1} x_{i4}
% + \sin(2 x_{i1}) + 0.5 \cos(3 x_{i2}) + 0.4 \cos(x_{i4})
% \big]\Big)}
% \end{aligned}
% \end{equation}
\begin{equation}
\small
\label{eq:nonlinear1}
\pi_i =
\frac{1}{
\begin{aligned}
&1 + \exp\Big(-\big[-1.2 + 0.6 x_{i1} - 1.5 x_{i2} + 1.7 x_{i3}
+ 1.5 x_{i4} - x_{i5} + 0.4 x_{i6} \\
&
+ 1.5 x_{i1}x_{i2}
+ 0.2 x_{i4}x_{i5}
+ 1.6 x_{i1}x_{i4}
+ \sin(2x_{i1})
+ 0.5\cos(3x_{i2})
+ 0.4\cos(x_{i4})
\big]\Big).
\end{aligned}
}
\end{equation}

% \begin{equation}
% \small 
% \label{eq:nonlinear2}
% \begin{aligned}
% \phi_i
% &= \frac{1}{1 + \exp\Bigl(-\bigl[0.6 + 0.5 z_{i1} +0.2 z_{i2} +0.6 z_{i3} +0.2 z_{i4}
% - 0.3 z_{i5} + 0.5 z_{i6} \\
% &\quad
% + 0.4 z_{i1} z_{i4} + 0.5 z_{i3} z_{i6} - z_{i2} z_{i3}
% + 0.6 \sin(1.5 z_{i3}) - 0.5 \sqrt{\left|z_{i6}\right|}
% + 2 \sin(z_{i1})
% \big]\Big)}.
% \end{aligned}
% \end{equation}
\begin{equation}
\small
\label{eq:nonlinear2}
\phi_i =
\frac{1}{
\begin{aligned}
&1 + \exp\Big(-\big[0.6 + 0.5 z_{i1} + 0.2 z_{i2} + 0.6 z_{i3}
+ 0.2 z_{i4} - 0.3 z_{i5} + 0.5 z_{i6} \\
&\quad
+ 0.4 z_{i1}z_{i4}
+ 0.5 z_{i3}z_{i6}
- z_{i2}z_{i3}
+ 0.6\sin(1.5z_{i3})
- 0.5\sqrt{\left|z_{i6}\right|}
+ 2\sin(z_{i1})
\big]\Big)
\end{aligned}
}.
\end{equation}
\iffalse
\begin{equation}
\label{eq:nonlinear}
\begin{aligned}
\pi_i
&= \frac{1}{1 + \exp\Big(-\big[-0.5 + 0.1 x_{i1} - 0.48 x_{i2} - 0.8 x_{i3}
+ 0.8 x_{i4} - 1.5 x_{i5} - 0.75 x_{i6} \\
&\quad
+ 1.5 x_{i3} x_{i2} + 0.2 x_{i4} x_{i3} + 1.6 x_{i2} x_{i4}
+ \sin(2 x_{i3}) + 0.5 \cos(3 x_{i2}) + 0.4 \cos(x_{i4})
\big]\Big)}, \\
\phi_i
&= \frac{1}{1 + \exp\Big(-\big[0.5 - 0.4 z_{i1} - z_{i2} - 0.1 z_{i3} - z_{i4}
- 0.3 z_{i5} + 0.5 z_{i6} + 0.7 z_{i7} \\
&\quad
- 2 z_{i1} z_{i4} + 1.8 z_{i3} z_{i6} - z_{i2} z_{i3}
+ 0.6 \sin(1.5 z_{i3}) - 0.5 \sqrt{\left|z_{i6}\right|}
+ 2 \sin(z_{i1})
\big]\Big)} .
\end{aligned}
\end{equation}
\fi 
%Regression coefficients for both components are pre-defined. 

Given covariate information, the presence probabilities are first generated (based on the functional forms in (\ref{eq:linear1}) for Case Study I and (\ref{eq:nonlinear1}) for Case Study II). A latent presence state is then generated according to the Bernoulli distribution in (\ref{eq:B1}). Conditional on the presence state, the detection probability is computed using the detection-related covariates (based on the functional forms in (\ref{eq:linear2}) for Case Study I and (\ref{eq:nonlinear2}) for Case Study II). Detection labels are then generated according to the Bernoulli distribution in (\ref{eq:B2}). If the species is absent, no detection is recorded. Note that during model fitting, only the detection labels and respective covariate information are ``seen'' by the model, as  %mimicking the information available 
in a real-world setting.

We run $100$ simulations, where for each simulation, we adopt a five-fold cross validation,  %Simulated observations are randomly partitioned into five folds, 
with four folds used for training and one fold for testing. %We then use detection labels and covariate information to train all models listed in Section \ref{subsec: exp}. Note that we differentiate between DeepZIB and DeepZIB-ST, based on whether or not a latent GRF is added in the presence function. 
To compute the AUC and $F_1$ score, we convert predicted probabilities into class labels (absent/present) through a threshold $\tau$,
determined for each method using Youden's J index %\citep{youden1950index}, 
such that     $\hat s_i =\mathbf{1}\!\left\{ \hat \pi_i \ge \tau \right\}$. %as in (\ref{eq:threshold}). 
%To ensure fairness, the threshold is determined for each method using Youden's J index \citep{youden1950index}, which identifies the point on the ROC curve that maximizes the separation between true positive rate and false positive rate. 
%\begin{equation}
%\label{eq:threshold}
%    \hat s_i =\mathbf{1}\!\left\{ \hat \pi_i \ge \tau \right\}.
%\end{equation}

%Table \ref{tb:eval-linear} and \ref{tb:eval-nonlinear} summarize the results under the linear and nonlinear generative scenarios. The threshold is then applied to convert predicted presence probabilities into binary predictions for evaluating classification metrics.

Table \ref{tb:eval-linear} reports the results under Case Study I (linear detection and presence functions). Here, it is evident that models accounting for imperfect detection outperform those assuming perfect detection. Although the ML models achieve moderate AUC values, they exhibit substantially higher Brier scores, NLL, and MAE, reflecting systematic mismatch induced by ignoring latent presence. Within the class of zero-inflated models, the ZIB model attains the best overall performance for Case Study I, %(linear presence and detection functions), 
which is expected given that the data-generating approach matches its parametric assumptions. %Importantly, the proposed DeepZIB variants 
The DeepZIB models achieve comparable but slightly degraded performance, %with slight degradation relative to the Classical ZIB model, 
indicating that the additional flexibility does not lead to overfitting and the models perform well even if underlying relationships are linear.
\begin{table}[htbp]
\caption{Case Study I (Linear scenario): performance comparison for the presence process across multiple benchmark models. Bold-faced values denote best performance. Numbers in parentheses are standard deviations around the average reported values. }
\label{tab:i}
%\textcolor{red}{[@Jiaxiang, please list brier and NLL in the last two columns. Put the threshold as the first column before AUC.]}}
\label{tb:eval-linear}
\centering
\setlength{\tabcolsep}{6pt}
\footnotesize
\resizebox{\textwidth}{!}{%
\begin{tabular}{lc|ccccc}
\toprule
Model & Threshold $\tau$ & AUC ($\uparrow$) &  $F_1$ Score ($\uparrow$) & MAE  ($\downarrow$) & Brier ($\downarrow$) & NLL ($\downarrow$)  \\
\midrule
LR        & .117 (.013)& .824 (.008)&  .572 (.015)& .135 (.003)& .145 (.004)& .469 (.012)\\
XGBoost   & .100 (.008)&.848 (.009)& .635 (.015)& .124 (.003)& .140 (.004)& .457 (.012)\\
MLP     & .467 (.029)& .864 (.008)& .654 (.017)& .133 (.009)& .125 (.005)& .400 (.012)\\
ZIB     & .252 (.044)& \textbf{.885} (.006)& \textbf{.687} (.015)& \textbf{.023} (.007)& \textbf{.096} (.003) & \textbf{.319} (.008)\\
DeepZIB     & .235 (.032)& .881 (.007)& .680 (.014)& .035 (.005)& .098 (.003)& .326 (.009)\\
DeepZIB-ST & .248 (.033)& .881 (.007)&  .682 (.014)& .036 (.006)& .098 (.003)& .326 (.009)\\
\bottomrule
\end{tabular}%
}
\end{table}

Table \ref{tb:eval-nonlinear} summarizes the results under Case Study II (non-linear detection and presence functions). In this case, the  ZIB model suffers substantial degradation across all metrics, reflecting its inability to capture the %non-linear 
now-more-complex covariate effects. %by its rather rigid parametric assumptions. 
XGBoost performs better than ZIB, even when trained directly on the detection labels. By contrast, %proposed DeepZIB and DeepZIB-ST 
DeepZIB models consistently achieve the best performance across all evaluation metrics, owing to %This is attributed to 
their ability to model the zero-inflated structure %of the data, as well as their flexibility to 
and to learn complex covariate effects. %Overall, these results confirm that properly accounting for zero inflation is essential, and that functional flexibility becomes critical when the true ecological relationships are complex.
\begin{table}[htbp]
\caption{Case Study II (Non-linear scenario): performance comparison for the presence process across multiple benchmark models. Bold-faced values denote best performance. Numbers in parentheses are standard deviations around the average reported values. }
\label{tab:ii}
%\textcolor{red}{[@Jiaxiang, please list brier and NLL in the last two columns. Put the threshold as the first column before AUC.]}}
\label{tb:eval-nonlinear}
\centering
\setlength{\tabcolsep}{3pt}
\footnotesize
\resizebox{\textwidth}{!}{%
\begin{tabular}{lc|ccccc}
\toprule
Model & Threshold $\tau$ & AUC ($\uparrow$) & $F_1$ Score ($\uparrow$) & MAE ($\downarrow$) & Brier ($\downarrow$) & NLL ($\downarrow$) \\
\midrule
LR        & .114 (.010) & .704 (.009)   &.472 (.011) &.167 (.002) & .177 (.004) &.573 (.013)\\
XGBoost   & .119 (.015)  & .764 (.011)  & .537 (.015) & .136 (.004) & .161 (.004) &.553 (.018)\\ 
MLP     & .557 (.040) & .798 (.011) & .572 (.017) & .162 (.008) & .160 (.005) & .485 (.013)\\
ZIB     & .380 (.065)  & .775 (.015) & .551 (.024) & .164 (.024) &.194 (.016) &1.005 (.301)\\
DeepZIB     & .256 (.032)  & \textbf{.849} (.007) & \textbf{.629} (.013) & \textbf{.050} (.005) & \textbf{.118} (.003) & \textbf{.382} (.009)\\
DeepZIB-ST & .293 (.031)  & \textbf{.849} (.007)  & \textbf{.629} (.013) & .055 (.006) & .119 (.004)&.386 (.010)\\
\bottomrule
\end{tabular}%
}
\end{table}
%In the linear scenario, the classical ZIB model performs slightly better than our proposed model, which is expected because the data-generation mechanism matches the model’s parametric assumptions. The proposed method, however, performs comparably and does not incur noticeable loss. In contrast, under the nonlinear scenario, the classical ZIB model suffers substantial degradation across nearly all metrics, reflecting its inability to capture the nonlinear covariate effects embedded in the generative mechanism. By comparison, the DeepZIB model attains the best performance across all the evaluation metrics. These gains highlight the flexibility of the deep neural network and its capacity to approximate nonlinear response surfaces while retaining the zero-inflated structure. 

\begin{figure}[tb]
\centering
\includegraphics[trim=2cm 2.5cm 3cm 5cm,clip,width = 0.9\linewidth]{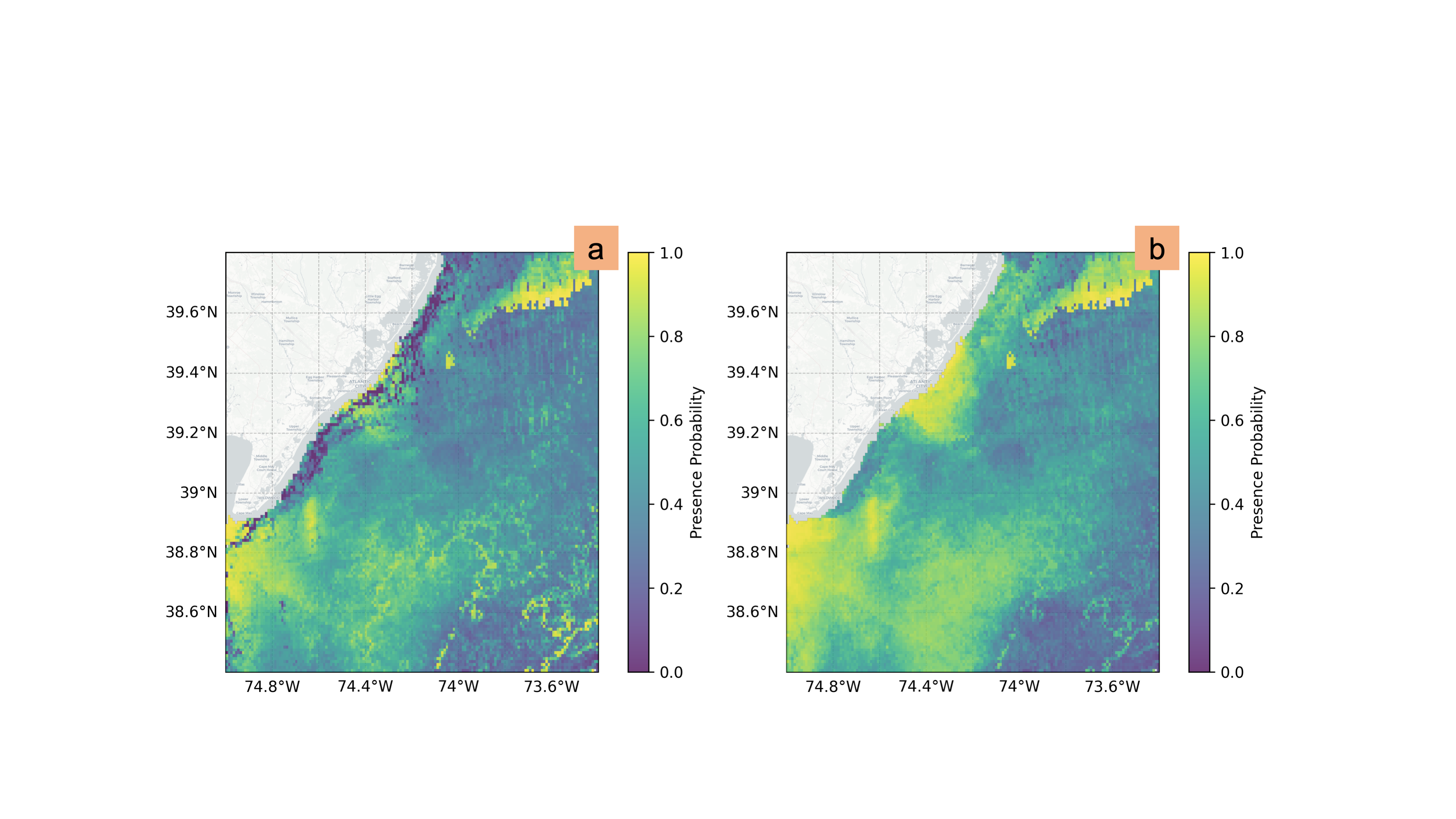}
\caption{Presence probability maps for the simulation study. (a): ground truth; (b) predicted using DeepZIB-ST.}%Left side is the ground truth probability map and right side is the predicted probability map for February 2022 using DeepZIB-ST.}
\label{fig:prob_map_sim}
\end{figure}

To further assess spatial prediction capability under the nonlinear simulation setting, Figure \ref{fig:prob_map_sim} compares the predicted presence probability map against the simulated ground truth. %The spatial covariate fields used to generate the ground‐truth presence probabilities are shown in Figure \ref{fig:covariates_map_sim}. 
Evaluating spatial accuracy is essential because a primary goal of our framework is to provide reliable extrapolations at new or unobserved locations, not only at sites where NARWs are observed, thereby %ensuring relevance to informing 
supporting blue economy management. Here, the model is trained on simulated observations through January 2022 and evaluated on simulated observations from February 2022. The resulting probability map in Figure \ref{fig:prob_map_sim}(b) closely reproduces both the high-probability hot spots and the broader spatial gradients present in the ground truth, demonstrating space-time generalization capabilities. %This demonstrates that the proposed method not only captures individual-level detection patterns but also generalizes well across space and time. 

To investigate the quality of the UQ, we construct %estimation uncertainty associated with the predicted latent presence probabilities using a nonparametric bootstrap. For each simulated training dataset, we generate $200$ bootstrap samples by sampling observations with replacement from the original training data. The models are then refitted to each bootstrap sample using the same architecture and configuration, yielding $200$ bootstrap estimates for each test observation. 
intervals at coverage levels of $10\%, 50\%,$ and $90\%$ using the bootstrapped empirical quantiles ($K = 200$). 
\iffalse
\begin{table*}[!htbp]
\centering
\caption{UQ results based on Winkler scores. Bold-faced values denote best performance.}  
%for the estimated latent presence probabilities. Lower values indicate better interval quality.}
\label{tab:uq_winkler}
\begin{tabular}{llcccccc}
\toprule
\textbf{}
& \textbf{}
& \textbf{XGBoost}
& \textbf{LR}
& \textbf{MLP}
& \textbf{ZIB}
& \textbf{DeepZIB}
& \textbf{DeepZIB-ST} \\
\midrule

\multirow{3}{*}{Case Study I}
& 10\% & 0.285 & 0.302 & 0.282 & 0.057 & 0.075 & 0.075 \\
& 50\% & 0.465 & 0.518 & 0.452 & 0.080 & 0.109 & 0.108 \\
& 90\% & 1.508 & 2.272 & 1.338 & 0.128 & 0.189 & 0.184 \\

\midrule

\multirow{3}{*}{Case Study II}
& 10\% & 0.311 & 0.373 & 0.367 & 0.322 & 0.107 & 0.107 \\
& 50\% & 0.520 & 0.653 & 0.601 & 0.530 & 0.156 & 0.156 \\
& 90\% & 1.911 & 2.951 & 2.045 & 1.917 & 0.273 & 0.272 \\

\bottomrule
\end{tabular}
\end{table*}
\fi
Figure \ref{fig:uq_tab} reports the Winkler scores across coverage levels and case studies. %Under the linear setting, the classical ZIB model achieves the lowest Winkler scores, while DeepZIB and DeepZIB-ST remain highly competitive. This finding is consistent with the results in Table \ref{tab:i}, as the linear data-generating mechanism closely matches the parametric assumptions of the classical ZIB model. Importantly, the comparable interval performance of the DeepZIB variants further suggests that their additional modeling flexibility does not substantially compromise uncertainty quantification when the underlying relationships are linear.
%Under the nonlinear setting, however, the relative performance changes substantially. DeepZIB and DeepZIB-ST achieve the lowest Winkler scores, whereas the classical ZIB model exhibits considerably poorer interval performance. This pattern is again consistent with the results in Table \ref{tab:ii}, highlighting the importance of flexible nonlinear modeling when the underlying presence and detection mechanisms are complex. 
The UQ results reinforce the conclusions of Tables \ref{tab:i} and \ref{tab:ii}: DeepZIB provides high-quality uncertainty intervals under simpler covariate effects (Case Study I), while unlocking clear advantages under more complex habitat relationships (Case Study II).
\begin{figure}[htbp]
\centering
\includegraphics[trim=0cm 3.5cm 3cm 4cm,clip,width = 0.9\linewidth]{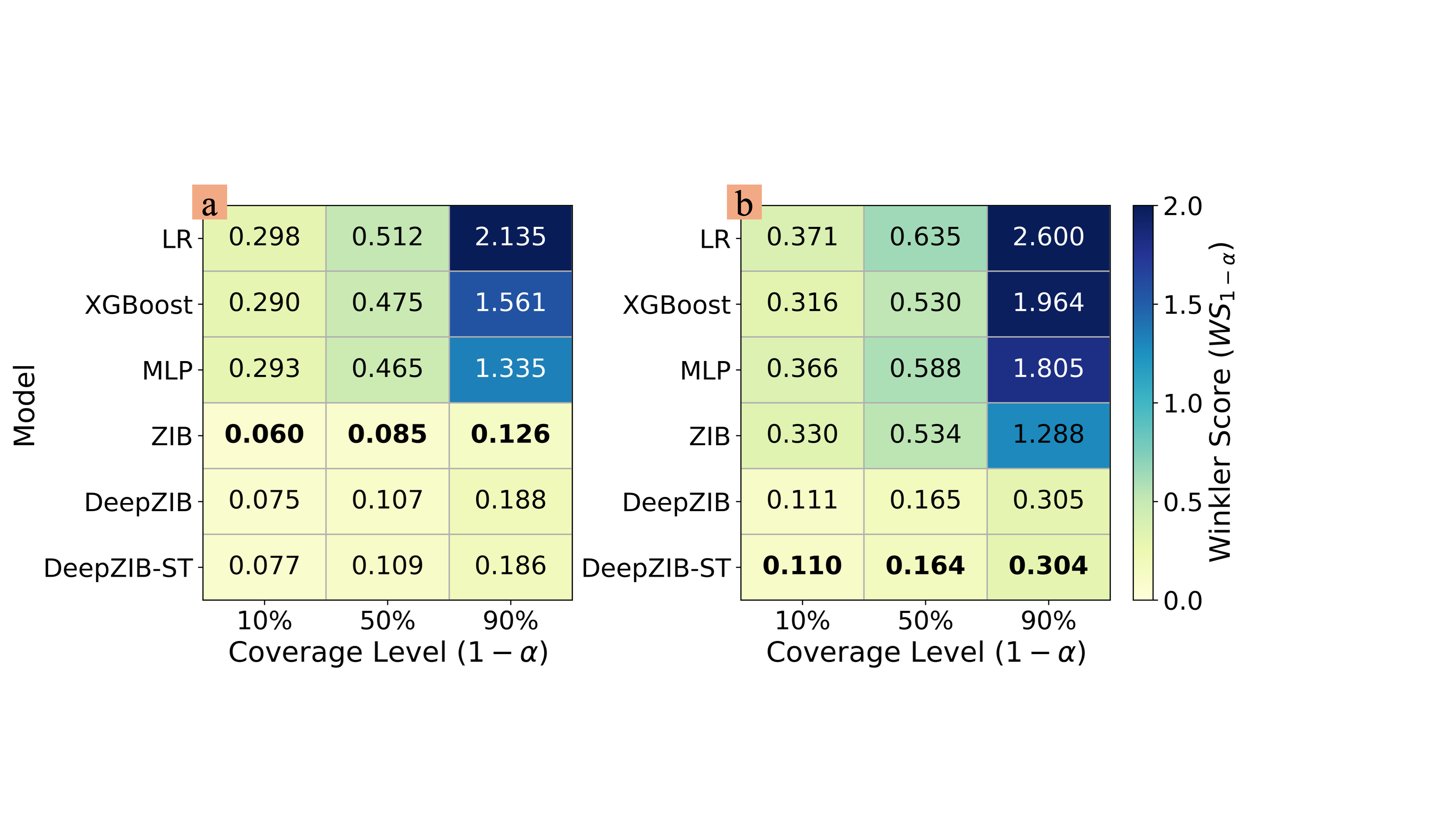}
\caption{UQ results based on Winkler scores. (a): Case study I; (b) Case study II. Lighter cell colors denote better performance. Bold-faced values denote best performance.}%Left side is the ground truth probability map and right side is the predicted probability map for February 2022 using DeepZIB-ST.}
\label{fig:uq_tab}
\end{figure}
\subsection{Modeling North Atlantic Right Whale Presence in the U.S. Mid-Atlantic and Connection to Blue Economy Management}
\label{sec: case}
%In this study, we integrate the vertically resolved glider- and spatially resolved satellite-based information for predicting the whale presence probability. 
%Building on the data source described in Section \ref{sec:data_description}, we conduct a real-world case study to evaluate the proposed model for predicting whale presence probability. In this application, satellite-derived environmental fields are used to characterize the habitat conditions associated with the latent presence process, while glider-based measurements inform the detection process along the deployment tracks.
%The data spans from Aug 2020 to June 2022, off the south coast of New Jersey between $38.5^\circ N$ to $39.5^\circ N$ and $74.5^\circ W$ to $73.5^\circ W$. Generation 3 Teledyne Slocum gliders with digital acoustic monitoring (DMON) instruments were used to detect the whale. 
%Specifically, habitat-related covariates include frontal value, sea surface temperature, chlorophyll-a concentration, distance to shore, and seasonal components represented by sine and cosine transformations of the month. Detection-related covariates are derived from the glider observations and include temperature, salinity, oxygen concentration, depth, elevation, and distance to shore. 
We use the data described in Section \ref{sec:data_description} to train our DeepZIB models. The presence and detection covariate vectors $\mathbf{x}$ and $\mathbf{z}$ are defined as in the simulation experiments. %Similar to the simulated experiments, the covariate vector for the presence process is defined as: $$\mathbf{x} = (\text{Frontal Value, Sea Surface Temperature, Chlorophyll, Distance to shore, Seasonal term}),$$ whereas the covariate vector for the detection process is defined as: $$\mathbf{z} = (\text{Depth, Temperature, Salinity, Oxygen concentration, Elevation, Distance to shore}).$$ 
The models are trained using all data prior to February 2022 and evaluated on data from February 2022 through September 2022. We set $R_s = 0.01^\circ\times 0.01^\circ$ ($\sim$ $1.24$ km$^2$) and $R_t=1$ month. %the information required for such an operational decision is different. 

Since the latent presence state cannot be observed directly, evaluating the model's performance relies on indirect evidence. %we therefore adopted multiple strategies to validate our model's performance. %First, because an acoustic detection (detection = 1) provides strong evidence of whale presence, we treat detection-positive observations as reliable presence data. 
We exploit the fact that acoustic detections have minimal false positives \citep{johnson2022acoustic}, that is, positive detections provide reliable, though incomplete, indicators of presence, whereas non-detections may reflect either true absence or imperfect detection. %a confirmed detection implies that at least one whale was present within the detection range of the glider. 
We therefore evaluate the model’s predictive performance using probabilistic metrics such as NLL and Brier score, which assess the quality of the predicted probability values. 
%Then, we compare our probability maps with independently derived density maps from other modeling frameworks, such as the density surface model. Although the scales and modeling approaches differ, we assess whether both models reveal similar spatial patterns. Such alignment would suggest that our model captures meaningful ecological signals.
Table \ref{tab:eval_transposed} presents the Brier Score and NLL for all benchmarks, evaluated exclusively on points with confirmed acoustic detections %(i.e., detection = 1), which serve 
as proxies for whale presence. Among the models, DeepZIB-ST achieves the best performance with the lowest NLL and Brier Score, followed by DeepZIB, and then ZIB as a distant third. In contrast, ML classifiers, which treat detection as direct evidence of presence without modeling latent states, yield significantly higher NLL and Brier Scores. These results align with the insights derived from the simulated case studies in Section \ref{sec: sim}. 
\begin{table}[htbp]
\centering
\caption{Test set evaluation for detected NARWs. Bold-faced values denote best performance. }%\textcolor{red}{[@Jiaxiang, are these results on the test set? If yes, please mention so.]}}
\label{tab:eval_transposed}
\begin{tabular}{|c|c|c|c|c|c|c|}
\hline
Metric & LR & XGBoost  & 
MLP &ZIB & DeepZIB & DeepZIB-ST \\ \hline
BS ($\downarrow$)  & 0.9728 & 0.9594        & 0.9479 & 0.8431 &     0.8415   & \textbf{0.6560} \\ \hline
NLL ($\downarrow$)  & 4.5313 &     3.8951    & 3.7549 & 3.9789 &     2.9129   & \textbf{1.8998} \\ \hline
\end{tabular}
\end{table}

Furthermore, we conduct an external validation by comparing the spatial patterns of our predicted presence probability maps with independently derived NARW density estimates from the DSM developed by the Marine Geospatial Ecology Lab at Duke University \citep{MGEL_OBIS_2022}. This comparison is not intended to quantify predictive accuracy of either model, as the two modeling frameworks rely on different data sources, assumptions, and responses. Specifically, the latter predicts the density of NARWs (individuals/unit area) and is trained on visual line-transect survey data. While it does not represent the golden ground truth, it is a long-standing and well-established model in the literature. %and has been evaluated independently on PAM data, showing positive correlations. %Instead, this comparison assess whether two independent approaches, informed by distinct observation processes, identify similar patterns. Together, these validation strategies provide evidence for the performance of the proposed models, despite the absence of direct presence observations.
The top row in Figure \ref{fig:narw_map} presents the spatial probability maps of NARW presence produced by DeepZIB-ST for February, March, April, and May 2022. A clear temporal pattern is observed: higher predicted probabilities are concentrated in February and March, followed by a gradual decline in the following months. The corresponding DSM-based maps are shown in the bottom row of Figure \ref{fig:narw_map}.  To produce such maps, we used the OBIS-SEAMAP mapping tool and only trimmed the region corresponding to our study area  \citep{MGEL_OBIS_2022}. Comparing the distribution and density maps in the top and bottom rows respectively, we find that DeepZIB-ST and DSM maps show broadly consistent spatial and temporal patterns, including the location and timing of hot spots and seasonal variation. In particular, as we approach summer, both the predicted probability (DeepZIB-ST) and the estimated density (DSM) are consistently decreasing, which may reflect the seasonal migration behavior of NARWs away from the study region \citep{whitt2013north}. For both sets of maps, there appears to be an increase in whale distribution (and density) as the distance from shore increases. The two approaches differ primarily in the spatial and temporal resolution at which NARWs are modeled. The DSM model is designed to model population-level density of NARWs over fairly large spatio-temporal scales, while our approach leverages where ML is most useful: to model fine-scale covariate effects at high resolution. %We discuss in Section \ref{sec: decisions} the relevance of such fine-grained predictions to localized decision support in the blue economy. %Our approach produces predictions at $1$-km resolution, and in principle, appears to carry more localized information content. %Specifically, the DSM output is provided at a coarse monthly temporal scale and spatially aggregated into large grid cells with a $5$ km resolution, whereas our model is fitted at approximately $1$ km resolution.
\begin{figure}[htbp]
\centering
\includegraphics[trim=0cm 5.5cm 0cm 0cm,clip,width = 1\linewidth]{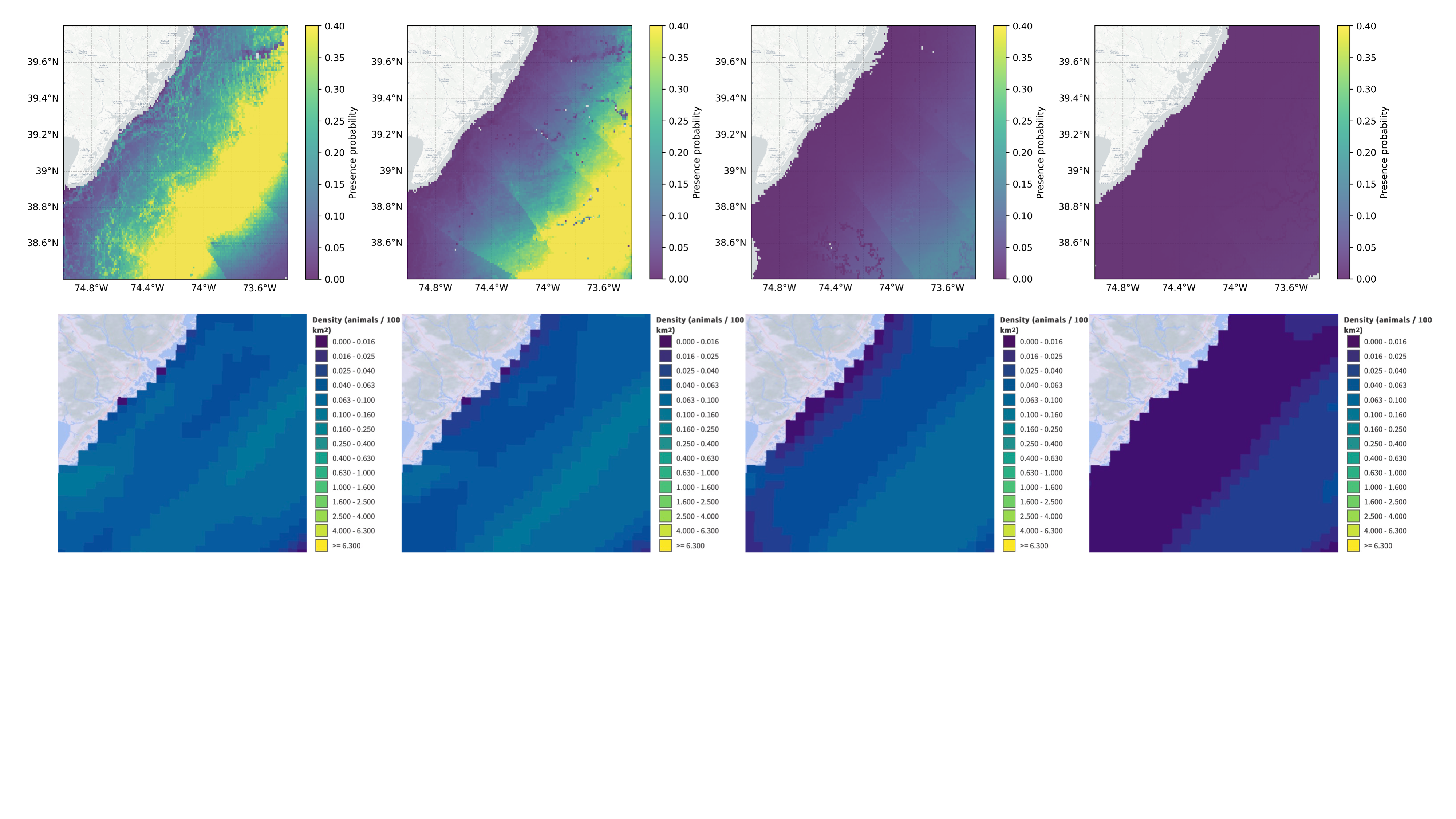}
\caption{Visual comparison between the spatial probability maps of NARW presence produced by DeepZIB-ST (top row) and the corresponding NARW density maps derived from the DSM of the Marine Geospatial Ecology Lab \citep{MGEL_OBIS_2022}, for February, March, April, and May 2022.}
%OBIS-SEAMAP tool. 

%\citep{MGEL_OBIS_2022}. The top row (left to right) presents the DeepZIB–estimated NARW presence probability maps for February–May 2022, while the bottom row (left to right) shows screenshots of }
\label{fig:narw_map}
\end{figure}

The importance of explicitly modeling imperfect detection is evident when comparing the outputs of the zero-inflated models versus those without a mechanism to account for zero inflation. Figure \ref{fig:showcase_three} shows the spatial probability maps of NARW presence generated using logistic regression (panel a), ZIB (panel b), and DeepZIB-ST (panel c) for February 2022. %All models have been trained on the same datasets. 
Here, logistic regression, which treats non-detections as absences, produces unrealistically flat probability maps, as it %This occurs because the model is forced to 
assumes all non-detections are ecological absences. In contrast, ZIB identifies regions of elevated presence probability, owing to its ability to recognize %even when detection data are sparse. 
%This confirms the importance of modeling 
imperfect detection. %suggests that imperfect detection must be modeled explicitly. 
%Beyond the benefit of the zero-inflated framework, 
The contribution of the deep component becomes apparent when comparing the spatial predictions from the ZIB model (panel b) with those from the DeepZIB-ST model (panel c). ZIB produces a fairly simplistic spatial pattern, with uniform spatial gradients, which is unlikely to reflect realistic NARW distributions. %most nearshore cells are assigned probabilities close to zero, and offshore cells are assigned high probabilities. 
In contrast, DeepZIB-ST shows significantly more localized spatial variations, possibly implying that the deep component is using the underlying environmental covariates to express complex, finer-scale habitat structure, which the linear ZIB formulation cannot capture. Again, this aligns well with the insights derived from our simulation experiments in Section \ref{sec: sim}, considering %where the classical ZIB models performed reasonably well in linear scenarios, but suffered substantial performance degradation in the non-linear case. The latter 
%where Case Study II appears to be closer to 
the ecological reality of NARWs and their complex habitats.  
\begin{figure}[tb]
\centering
\includegraphics[trim=0cm 6cm 0cm 4cm,clip,width = 0.9\linewidth]{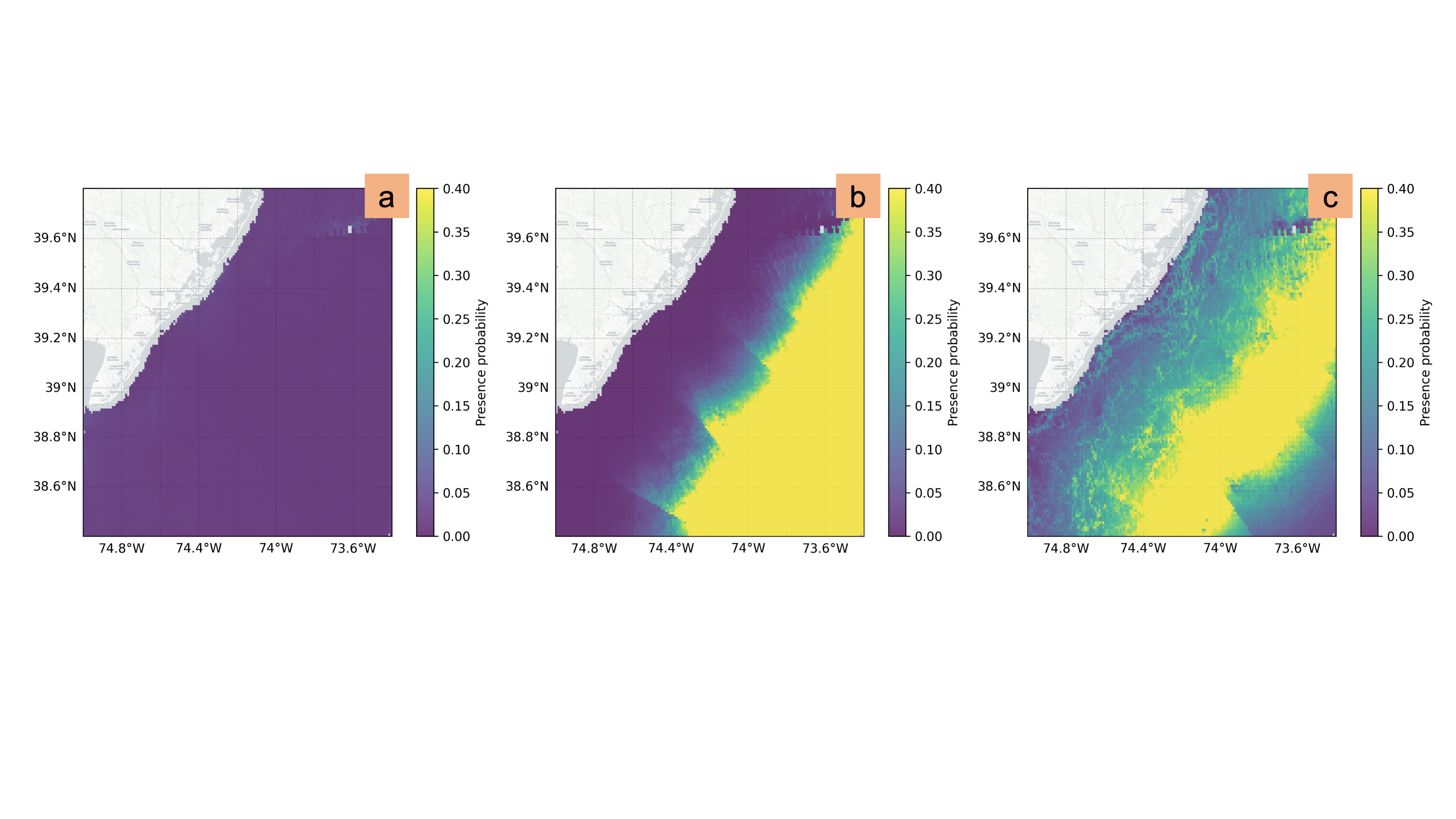}
\caption{Presence probability map produced by Logistic Regression (panel a), ZIB (panel b) and DeepZIB (panel c) on February 2022.}
\label{fig:showcase_three}
\end{figure}

To assess which covariates are most influential to drive the probability of presence, we conduct a leave-one-feature-out analysis and evaluate the change in NLL on the ground-truth presence set. Table \ref{tab:feature_importance_delta} %removing any covariate leads to a degradation in model fit, indicating 
suggests that all variables provide important habitat information, yet at different magnitudes. Removing distance to shore or seasonal terms led to the largest change in model fit, indicating that spatial patterns and seasonality (driven by month of the year) play a dominant role in predicting NARW presence. Oceanographic variables also contribute meaningfully, with the order of importance being: frontal value, SST, then chlorophyll.  
These findings are consistent with the habitat-related patterns of NARWs. In particular, the dominant influence of distance to shore and seasonal terms suggests that large-scale spatial and migratory patterns define the primary envelope of NARW presence in this geographical region. Within this envelope, localized oceanographic information plays an important role in shaping finer-scale habitat suitability. Among them, the frontal value emerges as a particularly informative predictor. It can be thought of as %aligning with its role as 
a proxy for oceanographic processes linked to prey aggregation \citep{baumgartner2003north,ji2024machine}. Meanwhile, SST and chlorophyll, which are commonly included as environmental covariates in NARW SDMs \citep{ garrison2012application}, exhibit lower, but still meaningful contributions. %\textemdash conclusions that align with prior NARW studies 

%suggesting that these variables fine-tune presence patterns instead of determining them outright.
%\textcolor{red}{[@Jiaxiang, Maybe a little more text here on the importance of these findings, similar to what we had in the Scientific reports paper?]}%while their effects are more localized, refining the shape of the probability surface.
\begin{table}[htbp]
\centering
\caption{Leave-one-feature-out contribution to the presence component. }
\label{tab:feature_importance_delta}
\begin{tabular}{lccc}
\toprule
\textbf{Model} & 
\textbf{NLL} & \textbf{$\Delta$NLL} & \textbf{$\Delta\pi$} \\
\midrule
All covariates for the presence process $\mathbf{x}$
& 1.8998 & -- & -- \\

$\mathbf{x}\setminus\{\text{Seasonal terms}\}$
& 3.2081 & +1.3083 & -0.1439 \\

$\mathbf{x}\setminus\{\text{Distance to shore}\}$
& 3.0000 & +1.1002 & -0.1395 \\

$\mathbf{x}\setminus\{\text{Frontal value}\}$
& 2.6475 & +0.7477 & -0.0860 \\

$\mathbf{x}\setminus\{\text{SST}\}$
& 2.4496 & +0.5498 & -0.0668 \\

$\mathbf{x}\setminus\{\text{Chlorophyll}\}$
& 2.3023 & +0.4025 & -0.0826 \\
\bottomrule
\end{tabular}
\end{table}

\subsection{Connection to Responsible Blue Economy Management} \label{sec: decisions}
The proposed framework offers operational advantages for the conservation of the critically endangered NARW in a growing Blue Economy space. %in settings where mitigation must be deployed in space and time rather than uniformly across an entire region. 
Dynamic ocean management was developed precisely for such contexts: rather than imposing static, coarse-scale mitigation boundaries (e.g., wide area closures for extended periods of time), %on mobile species and changing ocean conditions, 
it seeks to update management in response to changing biological and oceanographic information \citep{maxwell2015dynamic}. This ensures targeted mitigation without %For rare species such as NARWs, the core management problem is not whether to mitigate, but when and where to mitigate, so that protective measures can be targeted without 
imposing severe economic disruptions.

This need for spatially and temporally targeted information is already noticeable in current NARW mitigation efforts. For example, vessel speed rules are in place to enforce vessels to %65 feet or longer 
travel at 10 knots or less in designated seasonal management areas along the U.S. East Coast %, and those areas are defined as explicit geographic polygons or coastal bands 
\citep{NOAA_VesselStrikes_2025}. %NOAA also %implements a 500-yard no-approach safety zone around individual right whales and 
%frequently issues slow zones and dynamic management areas following visual or acoustic detections. 
In fisheries management, marine mammal protection includes area-based measures to reduce entanglement risks. %Taken together, these measures show that current NARW mitigation is an inherently spatio-temporal. %inherently spatially explicit, operating across multiple spatial scales from localized safe zones around individual whales to broader managed areas. 
Similar operational needs arise in the offshore energy sector, such as offshore wind farm construction, operations, and maintenance, where proactive decisions about vessel routing and scheduling can depend on spatially resolved information about marine species presence \citep{papadopoulos2021seizing}. %More broadly, the ramping blue economy activities will result in a congested marine transportation system that can significantly elevate the risk of vessel strikes and habitat disruption in proximity to marine infrastructure, as well as the transit to and from the port. 

DeepZIB is designed to produce outputs that can be directly paired with these mitigation policies. %By integrating satellite-based covariates with vertically resolved glider-based acoustic and environmental observations, the model
As seen in Figure \ref{fig:narw_map}, the model generates probability maps that retain local heterogeneity at scales relevant to spatial and temporal management of the blue economy. In this way, DeepZIB provides a principled pathway from high-frequency monitoring data to risk-aware information that can support high-resolution management. %Taken together, the DeepZIB framework supports dynamic and localized decision-making that is aligned with PAM observation processes and the operational constraints of Blue Economy sectors, providing a principled pathway from high-frequency monitoring data to actionable, risk-aware mitigation.
While DeepZIB model operates on a continuous spatial domain, the probability maps in this study are rendered on a grid that is as fine as $0.01^\circ$ (approximately 1 km). This resolution is bounded below by the spatial support of the environmental covariates, %used for prediction, 
but it is also scientifically meaningful, as it aligns with sub-meso-scale ocean variability and prey distributions. 
\iffalse
Ecologically relevant ocean structure can vary at submesoscale ranges, including horizontal scales on the order of about 1 km \citep{10.1098/rspa.2016.0117}. %Such fine-scale variability is important because it can organize marine ecosystems through localized gradients in environmental and prey conditions \citep{levy2018}. 
For NARWs, available evidence suggests that habitat use may involve multiple spatial scales: broader environmental cues may help identify suitable foraging habitat at scales of roughly 10--100 km, whereas prey-search behavior may occur at much finer scales of approximately 0.1 to 1 km \citep{sorochan2021}. 
\fi
Dynamic ocean management frameworks similarly emphasize management that adapts in space and time, refining the scale of managed areas in order to better balance conservation and economic objectives \citep{maxwell2015dynamic}.  In this context, the kilometer-scale spatial resolution of the DeepZIB framework is operationally meaningful, as it is well aligned with the broader transition toward spatially targeted, dynamic management strategies. 

To illustrate how the model outputs can be adapted to a specific operational decision, we consider a vessel dispatch decision, which is common in offshore asset construction and maintenance \citep{petersen2026data}. %in which vessel dispatch and maintenance activities must be scheduled over time. 
We set $R_s = 0.05^\circ\times0.05^\circ$ and $R_t = 1$ week, %offshore operating area and refine the temporal resolution to one week. The spatial resolution is coarsened because operational decisions such as vessel routing are typically made over an operating area rather than at the scale of individual kilometer-level prediction points, whereas the weekly temporal resolution better reflects the time scale at which vessel dispatches and maintenance activities may be planned and updated.
and provide a simple way to translate the probabilistic outputs from DeepZIB into actionable insights for ocean users. Specifically, we map the weekly probabilistic predictions, namely $\{\hat{\pi}_t, q_t^{0.05}, q_t^{0.95}\}$, into a decision-relevant state space $\mathcal{S}_t$ %that %\in {\text{Low probability with high confidence}, \text{Low probability but uncertain}, \text{High probability with high confidence}, \text{High probability but uncertain}}$, %based on the estimated presence probability and its uncertainty information, 
%reflects the magnitude and confidence of the predictions, 
as presented in (\ref{eq:uq_regime}). 
Here, 
%We use $\tau = 0.2$ as an illustrative threshold because the purpose of the analysis is to identify potential elevated NARW condition in a conservation sensitive setting. Let 
$\hat{\pi}_t$ is the predicted presence probability for week $t$, whereas $q_t^{0.05}$ and $q_t^{0.95}$ denote the corresponding $5$th and $95$th empirical quantiles, respectively. We set the threshold $c = 0.2$ to reflect a conservative setting typical of conservation planning of endangered species. %We classify the weekly condition according to (\ref{eq:uq_regime}).
\begin{equation}
\label{eq:uq_regime}
    \mathcal{S}_t =
\begin{cases}
\text{Low probability and high confidence}, & \hat{\pi}_t < c,\quad q_t^{0.95}<c,\\
\text{Low probability but uncertain}, & \hat{\pi}_t < c,\quad q_t^{0.95}\geq c,\\
\text{High probability and high confidence}, & \hat{\pi}_t \geq c,\quad q_t^{0.05}\geq c,\\
\text{High probability but uncertain}, & \hat{\pi}_t \geq c,\quad q_t^{0.05}< c.
\end{cases}
\end{equation}
%Regime A represents low predicted presence with high confidence; Regime B represents low predicted presence but high uncertainty; Regime C represents high predicted presence with high confidence; and Regime D represents high predicted presence but high uncertainty. 

Figure \ref{fig:regime} shows the resulting weekly states for the representative offshore site (around 50 km away from shore) from February through September 2022 (few weeks are missing due to insufficient data support). %Clearly, the operational state varies considerably over the study period. 
Predictions earlier in the year more frequently fall into elevated or uncertain regimes, whereas available summer predictions are mostly low and confident. % indicating that both the estimated presence probability and its upper uncertainty bound remain below the threshold. 
%Weeks shown in gray correspond to unavailable predictions and are not assigned to any regime. 
This information can be directly translated into a signal for vessel dispatch scheduling: These consistently low and confident summer conditions suggest that this period may provide a favorable scheduling window for heavier vessel traffic and offshore development activities. 

%for scheduling vessel dispatches and maintenance activities.
\begin{figure}[htbp]
\centering
\includegraphics[trim=0cm 0cm 0cm 0cm,clip,width = 1\linewidth]{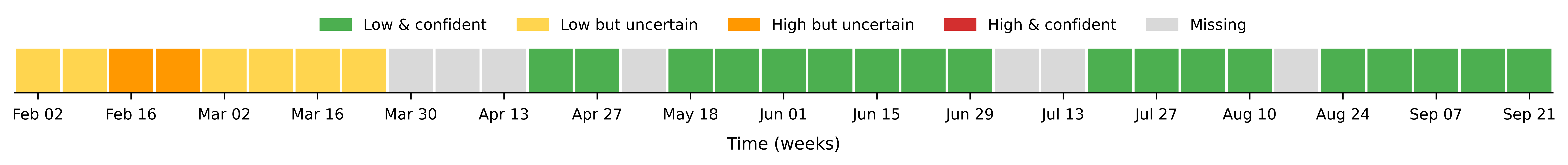}
\caption{Weekly NARW presence states over time. States are defined based on the predicted presence probabilities and their uncertainty information: low and confident, low but uncertain, high and confident, high but uncertain.}
\label{fig:regime}
\end{figure}

\section{Conclusion}
\label{sec:discussion}
In this paper, we proposed a Deep Zero-Inflated Bernoulli (DeepZIB) framework that integrates the statistical structure of zero-inflated models with the representational capacity of deep learning. Results from both simulation studies and a real-world case study demonstrate that this hybrid formulation yields systematic improvements over classical statistical SDMs and purely machine-learning-based benchmarks. A key advantage of DeepZIB is its ability to generate high-resolution, spatially and temporally varying presence maps, providing valuable insights for risk-aware management of blue economy industries.

Despite these promising results, a number of limitations point to important directions for future research. \textit{First,} while our model relies on satellite-derived oceanographic variables to infer whale presence, the spatial distribution of NARWs will be partly driven by prey availability. It is true that physical proxies such as SST, chlorophyll, and frontal value provide indirect information about habitat suitability. Hence, incorporating prey distribution is a natural step forward. \textit{Second,} although the detection process is explicitly modeled as a separate component, it is an approximation of a complex acoustic monitoring mechanism. Factors such as ambient noise and sound propagation substantially affect detection probability and are not fully captured by the current set of glider-based covariates. Extending the detection function in our model to better account for acoustic variability would further refine the separation between true ecological absence and observation failure.

\textit{Third,} the current formulation uses an MLP that processes covariate information in a pointwise manner. It %and 
does not fully exploit the spatial context of satellite imagery. Further work could explore advanced DNN architectures (e.g. convolutional, graph, or attention-based models) to extract spatial features from satellite imagery, enabling the model to capture spatial gradients that are difficult to represent using pointwise covariates. This naturally complicates parameter learning and requires careful methodological and algorithmic solutions. 
\textit{Finally,} an important avenue for future research lies in integrating the probabilistic model outputs into operational decision support systems to advance responsible blue economy development and operation. Section \ref{sec: decisions} showed a simple path forward. Future work will examine how the predicted presence probability maps can be rigorously incorporated into risk-aware dynamic planning in offshore energy operations, vessel scheduling, fisheries management, and maritime transport. This would support mitigation strategies that balance conservation objectives with human activities in the context of the growing blue economy.

\section*{Supplementary Material}

% Reset numbering
\setcounter{section}{0}
\renewcommand{\thesection}{S.\arabic{section}}
\renewcommand{\theHsection}{supplement.\arabic{section}}
Section \ref{sec:proof} provides the proof of the identifiability result (Lemma 1) stated in the main text. 
Section \ref{sec:metrics} provides definitions of the evaluation metrics used in the manuscript.
\section{Proof of Lemma 1 (the Identifiability Result)}
\label{sec:proof}
This section provides the proof of the identifiability result stated in Lemma 1 of the main manuscript. 
For completeness, we first restate the result below. 
The notation follows that introduced in the main text.

\begin{lemmaone}
Under the assumption that there exists a continuous variable $z_\ell\in\mathcal Z\setminus\mathcal X$
such that $\partial_{z_\ell}\log\varphi(\mathbf{z};\boldsymbol \psi)$ exists,
the Deep ZIB model is identifiable if $Pr(Y=y|\mathbf{x},\mathbf{z};\boldsymbol \vartheta) = Pr(Y=y|\mathbf{x},\mathbf{z};\boldsymbol \vartheta^*)$ implies $\pi(\mathbf{x};\boldsymbol{\Theta}) = \pi(\mathbf{x};\boldsymbol{\Theta^*})$ $\forall \mathbf{x}\in \mathcal{X}$, $\varphi(\mathbf{z};\boldsymbol \Psi) = \varphi(\mathbf{z};\boldsymbol \Psi^*)$ $\forall \mathbf{z} \in \mathcal{Z}$
\end{lemmaone}

\begin{proof}
Recall that under the DeepZIB model,
\[
\Pr(Y=1 \mid \mathbf{x}, \mathbf{z}; \boldsymbol{\vartheta})
= \pi(\mathbf{x}; \boldsymbol{\Theta})\,\varphi(\mathbf{z}; \boldsymbol{\Psi}),
\]
and $\Pr(Y=0 \mid \mathbf{x}, \mathbf{z}; \boldsymbol{\vartheta}) = 1 -
\pi(\mathbf{x}; \boldsymbol{\Theta})\,\varphi(\mathbf{z}; \boldsymbol{\Psi})$.
Therefore, equality of the conditional distributions for $Y$ is equivalent
to equality of $\Pr(Y=1 \mid \mathbf{x}, \mathbf{z})$, and it suffices to
consider the case $y=1$.

Suppose that for all $(\mathbf{x},\mathbf{z}) \in \mathcal{X}\times\mathcal{Z}$,
\[
\pi(\mathbf{x};\boldsymbol{\Theta})\,\varphi(\mathbf{z};\boldsymbol{\Psi})
=
\pi(\mathbf{x};\boldsymbol{\Theta}^*)\,\varphi(\mathbf{z};\boldsymbol{\Psi}^*).
\]
Taking logarithms of both sides yields
\begin{equation}
\label{eq:logsep}
\log \pi(\mathbf{x};\boldsymbol{\Theta}) - \log \pi(\mathbf{x};\boldsymbol{\Theta}^*)
=
\log \varphi(\mathbf{z};\boldsymbol{\Psi}^*) - \log \varphi(\mathbf{z};\boldsymbol{\Psi}).
\end{equation}
The left-hand side of \eqref{eq:logsep} depends only on $\mathbf{x}$, whereas the
right-hand side depends only on $\mathbf{z}$. By assumption, there exists a
continuous variable $z_\ell \in \mathcal{Z}\setminus\mathcal{X}$ such that
$\partial_{z_\ell}\log \varphi(\mathbf{z};\boldsymbol{\Psi})$ exists.
Differentiating both sides of \eqref{eq:logsep} with respect to $z_\ell$ therefore
yields
\[
\frac{\partial}{\partial z_\ell}
\Bigl(
\log \varphi(\mathbf{z};\boldsymbol{\Psi}^*)
-
\log \varphi(\mathbf{z};\boldsymbol{\Psi})
\Bigr)
= 0
\quad \text{for all } \mathbf{z} \in \mathcal{Z}.
\]
It follows that $\log \varphi(\mathbf{z};\boldsymbol{\Psi}^*) -
\log \varphi(\mathbf{z};\boldsymbol{\Psi})$ is constant on $\mathcal{Z}$, and hence
there exists a constant $C>0$ such that
\[
\varphi(\mathbf{z};\boldsymbol{\Psi}^*) = C\,\varphi(\mathbf{z};\boldsymbol{\Psi})
\quad \text{for all } \mathbf{z}\in\mathcal{Z}.
\]
Substituting this relation back into the equality of probabilities implies
\[
\pi(\mathbf{x};\boldsymbol{\Theta}) = C\,\pi(\mathbf{x};\boldsymbol{\Theta}^*)
\quad \text{for all } \mathbf{x}\in\mathcal{X}.
\]

To determine the value of $C$, note that $\pi(\cdot)$ and $\varphi(\cdot)$ are
probability-valued functions taking values in $(0,1)$. Since the neural network
parameterizations are unconstrained apart from the sigmoid output, there exist
sequences $\{\mathbf{x}_n\}\subset\mathcal{X}$ and $\{\mathbf{z}_n\}\subset\mathcal{Z}$
such that $\pi(\mathbf{x}_n;\boldsymbol{\Theta}) \to 1$ and
$\varphi(\mathbf{z}_n;\boldsymbol{\Psi}) \to 1$. The above relations then imply
$\pi(\mathbf{x}_n;\boldsymbol{\Theta}^*) \to 1/C$ and
$\varphi(\mathbf{z}_n;\boldsymbol{\Psi}^*) \to C$, which is only possible if
$C \le 1$ and $C \ge 1$. Hence $C=1$.

Therefore,
\[
\pi(\mathbf{x};\boldsymbol{\Theta}) = \pi(\mathbf{x};\boldsymbol{\Theta}^*)
\quad \forall\,\mathbf{x}\in\mathcal{X},
\qquad
\varphi(\mathbf{z};\boldsymbol{\Psi}) = \varphi(\mathbf{z};\boldsymbol{\Psi}^*)
\quad \forall\,\mathbf{z}\in\mathcal{Z},
\]
which establishes identifiability of the DeepZIB model.
\end{proof}

\section{Evaluation Metrics}
\label{sec:metrics}
This section provides formal definitions of the evaluation metrics used to assess the predictive performance and uncertainty quantification of the proposed DeepZIB framework. 
For the simulation study, the true latent presence indicators and presence probabilities are available, allowing evaluation of both classification performance and probability estimation accuracy. 
For uncertainty quantification, the estimated intervals are additionally evaluated in terms of the Winkler score.

Let $s_i\in\{0,1\}$ denote the true latent presence state for sample $i$, 
$\pi_i$ the corresponding true presence probability, and 
$\hat{\pi}_i$ the estimated presence probability. 
Let $\hat{s}_i$ denote the predicted binary presence state obtained by thresholding $\hat{\pi}_i$. For uncertainty quantification, let $[L_i^{(\alpha)},U_i^{(\alpha)}]$ denote the $(1-\alpha)$ uncertainty interval for $\pi_i$. 
The following metrics are used to evaluate model performance.
\begin{enumerate}

    \item 
%\item{(\textit{i}) 
Area under the Receiver Operating Characteristic (ROC) curve (AUC),
%which measures a model’s ability to rank present samples above absent ones, 
as defined in (\ref{eq: auc}), where $\mathcal{D}_1, \mathcal{D}_2$ denote the sets of presence and absence samples. Higher values of AUC indicate better performance in distinguishing the presence and absence states. 
\begin{equation}
    \label{eq: auc}
    \mathrm{AUC} = \frac{\sum_{i \in \mathcal{D}_1}\sum_{j \in \mathcal{D}_2}\mathbf{1}\!\left( \hat{\pi}_i > \hat{\pi}_j \right)}{|\mathcal{D}_1||\mathcal{D}_2|}.
\end{equation}

%\item{(\textit{ii})
\item 
\texorpdfstring{$F_1$} ~ score, which is the harmonic mean of precision and recall, as expressed in (\ref{eq:f1}). Higher values indicate better classification performance.
\begin{equation}
\label{eq:f1}
\mathrm{F_1}
=
\frac{
  2 \times \sum_{i=1}^n \mathbf{1}(s_i = 1, \hat{s}_i = 1)
}{
  \sum_{i=1}^n \mathbf{1}(s_i = 1) + \sum_{i=1}^n \mathbf{1}(\hat{s}_i = 1)
}.
\end{equation}

%\item{(\textit{iii}) 
\item Mean absolute error (MAE) %which computes the mean squared difference 
measuring the divergence between the predicted and true probabilities, as expressed in (\ref{eq:mae}). 
\begin{equation}
\label{eq:mae}
\mathrm{MAE}
=
\frac{1}{n}
\sum_{i=1}^{n}
\bigl|
  \hat{\pi}_i - \pi_i
\bigr|.
\end{equation}

%\item{(\textit{iv}) 
\item Brier score which measures the difference between $\hat{\pi}_i$ and the binary presence indicators $s_i$, as defined in (\ref{eq:bs}). Smaller values of the Brier score correspond to better performance.
\begin{equation}
\label{eq:bs}
\mathrm{Brier~score}
=
\frac{1}{n}
\sum_{i=1}^{n}
\bigl( \hat{\pi}_i - s_i \bigr)^2.
\end{equation}

%\item{(\textit{v}) 
\item Negative log-likelihood (NLL), which measures the adequacy of the model, as expressed in (\ref{eq:nll}). Lower values of NLL indicates a more adequate model. %better-calibrated probabilistic predictions.
\begin{equation}
\mathrm{NLL}
=
-
\sum_{i=1}^{n}
\Bigl[
  s_i \log(\hat{\pi}_i)
  +
  (1-s_i)\log\bigl(1-\hat{\pi}_i\bigr)
\Bigr].
\label{eq:nll}
\end{equation}
\item \textbf{Winkler score.}
The Winkler score evaluates the quality of uncertainty intervals by jointly accounting for interval width and whether the true probability is covered by the interval. 
For a $(1-\alpha)$ uncertainty interval $[L_i^{(\alpha)},U_i^{(\alpha)}]$, the Winkler score for sample $i$ is defined as (\ref{eq:winkler}).
\begin{equation}
\label{eq:winkler}
W_{i,\alpha}
=
\begin{cases}
U_i^{(\alpha)}-L_i^{(\alpha)},
&
L_i^{(\alpha)}
\leq
\pi_i
\leq
U_i^{(\alpha)},\\
U_i^{(\alpha)}-L_i^{(\alpha)}
+
\displaystyle\frac{2}{\alpha}
\left(
L_i^{(\alpha)}-\pi_i
\right),
&
\pi_i<L_i^{(\alpha)},\\[8pt]
U_i^{(\alpha)}-L_i^{(\alpha)}
+
\displaystyle\frac{2}{\alpha}
\left(
\pi_i-U_i^{(\alpha)}
\right),
&
\pi_i>U_i^{(\alpha)}.
\end{cases}
\end{equation}
The reported score is averaged across all test samples, as expressed in (\ref{eq:mean_winkler}). Smaller values indicate a more favorable combination of interval sharpness and coverage. 
\begin{equation}
\label{eq:mean_winkler}
\mathrm{WS}_{1-\alpha}
=
\frac{1}{n}
\sum_{i=1}^{n}
W_{i,\alpha}.
\end{equation}

\end{enumerate}

\bibliographystyle{imsart-nameyear} % Style BST file
\bibliography{sample} 
%% update the figure
\end{document}